\documentclass[11pt]{article}
\usepackage{color}
\usepackage{authblk}
\usepackage{setspace}
\usepackage{amssymb}   %basic
\usepackage{amsthm}    %theoremstyle
\usepackage{amsmath}   %\mathcal
\usepackage{mathtools}
\usepackage{stmaryrd}  %\interleave
\usepackage{titletoc}  %\ttl
\usepackage{mathrsfs}  %\mathscr
\usepackage{graphicx}
\usepackage{todonotes}
\usepackage{enumerate}
\usepackage{floatrow}
\usepackage{multirow}
\usepackage{multicol}
\usepackage{enumerate}
\usepackage{caption}
\usepackage{subcaption}

\newlength{\defbaselineskip}
\newcommand{\setlinespacing}[1]%
           {\setlength{\baselineskip}{#1 \defbaselineskip}}
\theoremstyle{plain}
\newtheorem{thm}{Theorem}[section]

\newtheorem{lem}[thm]{Lemma}

\newtheorem{exam}[thm]{Example}

\theoremstyle{definition}

\newtheorem{ass}{Assumption}[section]
\newtheorem{rmk}{Remark}[section]

\DeclareMathOperator*{\esssup}{esssup}
\DeclareMathOperator*{\essinf}{essinf}

\newcommand{\cN}{\mathcal{N}}
\newcommand{\cB}{\mathcal{B}}

\newcommand{\cE}{\mathcal{E}}

\newcommand{\cF}{\mathcal{F}}

\newcommand{\cW}{\mathcal{W}}

\newcommand{\cX}{\mathcal{X}}
\newcommand{\cY}{\mathcal{Y}}
\newcommand{\cZ}{\mathcal{Z}}

\newcommand{\bE}{\mathbb{E}}
\newcommand{\bP}{\mathbb{P}}
\newcommand{\bR}{\mathbb{R}}
\newcommand{\bN}{\mathbb{N}}

\newcommand{\sF}{\mathscr{F}}
\newcommand{\sP}{\mathscr{P}}

\DeclareMathOperator{\var}{var}
\newcommand{\abs}[1]{\left|{#1}\right|}

\makeatletter\@addtoreset{equation}{section} \makeatother

 \allowdisplaybreaks
\begin{document}

\title{ A Deep Learning-Based Method for Fully Coupled Non-Markovian FBSDEs with Applications
}

\author[3]{Hasib Uddin Molla\thanks{Corresponding author. Email: \texttt{mdhasibuddin.molla@ucalgary.ca}}}
\author[3] {Matthew Backhouse }
\author[3]{Ankit Banarjee}

\author[3] {Jinniao Qiu}

% \affil[1]{School of Mathematics,
% 	University of Edinburgh}
% \affil[2]{	School of Mathematics and Statistics,
% 	University of New South Wales}

\affil[3]{Department of Mathematics and Statistics,
	University of Calgary}

%\date{} % Optional: removes date

\maketitle

\begin{abstract} 
In this work, we extend deep learning-based numerical methods to fully coupled forward-backward stochastic differential equations (FBSDEs) within a non-Markovian framework. Error estimates and convergence are provided. In contrast to the existing literature, our approach not only analyzes the non-Markovian framework but also addresses fully coupled settings, in which both the drift and diffusion coefficients of the forward process may be random and depend on the backward components $Y $and $Z$. Furthermore, we illustrate the practical applicability of our framework by addressing utility maximization problems under rough volatility, which are solved numerically with the proposed deep learning-based methods.\\

\noindent \textbf{Keywords:} .
\end{abstract}

\section{Introduction} 
Let $(\Omega, \sF, (\sF_t)_{t\in[0,T]}, \bP)$ be a complete filtered probability space, where the filtration $(\sF_t)_{t\in[0,T]}$ is the augmented filtration generated by an $m$-dimensional Wiener process $(W_t)_{t\in[0,T]}$. The predictable $\sigma$-algebra on $\Omega \times [0,T]$ associated to $(\sF_t)_{t\in[0,T]}$ is denoted by $\sP$.
\\[4pt]
We consider the following coupled forward-backward stochastic differential equation (FBSDE):
\begin{equation}\label{maineq}
    \begin{cases}
        dX_t &= b(t,V_t,X_t,Y_t,Z_t)dt + \sigma(t,V_t,X_t,Y_t,Z_t)dW_t, \ t\in[0,T],\\
        X_0&=x_0, \\
        -dY_t &=f(t,V_t,X_t,Y_t,Z_t)dt - Z_t dW_t, \ t\in[0,T],\\
        Y_T &= g(V_T,X_T).
    \end{cases}
\end{equation}
where $T>0$ is a fixed finite terminal time, and the solution triple $(X_t,Y_t,Z_t)$ takes values in $\bR^d\times\bR^{d_0}\times\bR^{d_0\times m}$. The functions $b,\sigma,f$, and $g$ take values in appropriate spaces consistent with the dimensions of the problem. Wiener process $W_t$ is decomposed as $W_t=(\widetilde{W}_t,B_t)$, where $\widetilde{W}_t$ is $m_1$-dimensional and $B_t$ is $m_2$-dimensional, with $m_1+m_2=m$. When $m_1=0$ or $m_2=0$, we mean $W=B$ or $W=\widetilde{W}$ respectively. \\[4pt]
The following assumption is imposed on the exogenous stochastic process $V$.
\begin{ass}
  \label{ass:stoch-var}
  The process $V$ has continuous trajectories, takes values in $\bR^{m_0}$, and is adapted to the filtration generated by the Wiener process $\widetilde{W}$. In addition, it is integrable in the sense that
  \begin{equation*}
    \bE\left[ \int_0^T |V_s| ds \right] < \infty, \quad T > 0.
  \end{equation*}
\end{ass}
\noindent
 It is important to emphasize that $V$, and even the joint process $(X, V)$, is not assumed to be Markovian or to possess the semi-martingale property. In fact, neither of our main examples satisfies these properties. 
 \\[4pt]
 FBSDEs have found numerous applications in modelling optimization problems in mathematical finance and have been the subject of growing interest over the past three decades. Considerable effort has been devoted to understanding the solvability of fully coupled FBSDEs, particularly over arbitrary time horizons and under minimal regularity assumptions on the coefficients.
One of the earliest approaches is the contraction mapping technique introduced by Antonelli \cite{FAntonelli1993BFSDEs}, which establishes the existence and uniqueness of solutions but only for sufficiently small time intervals. The four-step scheme developed by Ma et al. \cite{JMaPProtterJYong1994SolvingFBSDEs4StepScheme} provides a more direct method of solution by linking the problem to a system of parabolic partial differential equations (PDEs). Nonetheless, this method therein is restricted to the Markovian setting. Another approach is the so-called continuation method, proposed by Hu and Peng \cite{YHuSPeng1995SolFBSDEs} and further developed by Yong \cite{JYong1997FindAdaptedSolFBSDEMethodofCont}, which accommodates non-Markovian frameworks and is applicable over general time horizons. This method has proven effective in broadening the class of FBSDEs for which well-posedness can be established.
 \\[4pt]
 Under the non-Markovian framework, the resolution of an FBSDE  is closely linked to the existence of a random field $u(x,t)$ such that $Y(t)=u(X(t),t),t\in[0,T]$. This random field, known as the decoupling field of the FBSDE, satisfies an associated quasilinear backward stochastic partial differential equation (BSPDE). The connection between FBSDEs and BSPDEs, as well as the conditions under which they are well-posed, has been investigated in \cite{JMaHYinJZhang2012OnNon-MarkovFBSDEsandBSPDEs, JMaZWuDZhangJZhang2015OnWellPossedofFBSDEsAUnifiedAppr}; 
 in the Markovian case, the decoupling field becomes deterministic and satisfies a deterministic quasilinear parabolic PDE. This correspondence between FBSDEs and BSPDEs forms the foundation of the stochastic Feynman–Kac formula, which provides a probabilistic representation for the solution of certain classes of partial differential equations.
 \\[4pt]
 Over the past decades, researchers have developed several approaches for numerical approximations of FBSDEs. Some of these approaches exploit the connections between FBSDEs and quasilinear parabolic partial differential equations (PDEs), particularly through the four-step scheme. These methods rely on the numerical solutions of parabolic PDEs,  for example through finite difference techniques, which typically become computationally prohibitive in high-dimensional contexts due to the curse of dimensionality. 
 To overcome the limitations of PDE-based approaches, alternative methodologies based on time discretization schemes have been explored. For decoupled FBSDEs,  the discretization of the forward component is relatively straightforward, whereas the backward component involves the computation of conditional expectations at each time step--a task that poses significant challenges, particularly in high-dimensional settings. Such time discretization-based schemes for decoupled FBSDEs are studied by Zhang \cite{JZhang2004NumericalSchemeBSDEs}, while regression-based Monte-Carlo methods were proposed by Bouchard and Touzi \cite{BBouchardNTouzi2004DiscreteTimeApprxMCsimulBSDE}, and further developed by Gobet et al. \cite{EGobetJPLemorXWarin2005RegressionBasedMCtoSolveBSDEs} among many others. 
Despite ongoing efforts to reduce the computational complexity associated with evaluating conditional expectations, this remains a significant obstacle for high-dimensional problems. In the case of coupled FBSDEs, these difficulties are exacerbated, as time discretization schemes often require iterative procedures such as Picard iterations, which further increase computational cost. To address this, Bender and Zhang \cite{CBenderJZhang2008TimeDiscrMarkovIteraCoupledDFBSDEs} proposed a time discretization scheme incorporating Markovian iteration to mitigate the complexity inherent in solving coupled FBSDEs numerically.
 \\[4pt]
 Since most numerical algorithms for FBSDEs and parabolic PDEs suffer from the so-called curse of dimensionality, which renders them inefficient or infeasible in high-dimensional settings, a class of deep learning-based algorithms has recently emerged to overcome this challenge. Notably, several methods have been developed for decoupled FBSDEs and their associated parabolic PDEs \cite{EWeinanJHanAJent2017DeepLearnNumrMethodHighDimParbPDEsBSDEs,JHanAJentWE2018SolvingHighDimEqusDeepLearnin,CHureHPhamXWarin2020DeepBShcemeHighDimNonlPDEs}, demonstrating the capacity to mitigate the challenges associated with high-dimensional problems.
 These deep learning-based approaches, commonly referred to as Deep BSDE methods, have since been extended to some coupled FBSDEs \cite{JHanJLong2020ConvgDeepBSDEforCoupledFBSDEs,SJiSPengYPengXZhang2020ThreeAlgorithmsforSolvHighDimCoupledFBSDEsDeepLearn,HUMollaJQiu20221NumercFBSDE}. In particular, Han and Long \cite{JHanJLong2020ConvgDeepBSDEforCoupledFBSDEs} established a convergence analysis and obtained a posteriori error estimates for the Deep BSDE scheme in the coupled setting. However, these studies have largely been confined to the Markovian framework and to cases where the coefficients of the forward SDE do not depend on the unknown process $Z$. More recently, Bayer et al. \cite{CBayerJQiuYYao2020PricingOptnRoughVolatBSPDEs} extended the application of Deep BSDE methods to a class of FBSDEs under a non-Markovian framework but restricted to a \textit{depcoupled} setting.
\\[4pt]
We emphasize that the FBSDEs considered in this work are non-Markovian, in the sense that their coefficients may exhibit randomness through dependence on external non-Markovian processes, such as $V$ in our case. However, this form of randomness, arising from an exogenous process, is distinct from that in path-dependent FBSDEs, where the coefficients depend explicitly on the entire trajectory of the forward process $X$. Such path-dependent FBSDEs typically arise in stochastic control problems with memory or path-dependent features, including applications to exotic derivatives such as Asian, barrier, or lookback options. These classes of problems are also naturally connected to the theory of path-dependent partial differential equations (PPDEs) (see \cite{TPhamJZhang_2014_ZeroSumGamePathDependent},\cite{Saporito2019StochControlPathDependent} for example). Neural network-based methods \cite{SaporitoZhang2021NeuralNetPathDependentPDE},\cite{QiFeng_2023_DeepSignatureFBSDEAlgo} can be used for solving path-dependent partial differential equations (PPDEs) and FBSDEs.

In this work, we investigate Deep BSDE methods for fully coupled non-Markovian FBSDEs. Error estimates are obtained for a class of coupled FBSDEs in which both the drift and diffusion components of the forward process may be random and depend on the backward process $Z$ and $Y$. The remainder of this paper is organized as follows. In Section 2, we introduce the necessary notations and assumptions, review key properties of solutions to FBSDEs, and provide a brief overview of neural network approximation techniques. Section 3 presents the framework of our proposed numerical scheme and derives an estimate for the associated simulation error. In Section 4, we carry out the convergence analysis of the scheme. Section 5 details the algorithm for the numerical approximation of fully coupled non-Markovian FBSDEs. Finally, in Section 6, we illustrate the effectiveness of the proposed method through two numerical examples.

%%%%%%%%%%%%%%%%%%%%%%%%%%%%%%%%%%%%%%%%%%%%%%%%%%%%%%%%%%%%%%%%%%%
\section{Preliminaries}
\subsection{Notations and Assumptions}

In this section, we introduce the notations and assumptions used throughout the paper. Let $\pi:0=t_0<t_1<t_2\cdots<t_{N-1}<t_N=T $ denote a partition of the time interval $[0,T]$ into $N$ subintervals $(t_{i},t_{i+1})$. We define the mesh size of the partition as 
$$h:=\max\{t_{i+1}-t_i:i=0,1,\cdots,N-1\}.$$
Without loss of generality, we assume that the partition is uniform, i.e., $h=t_{i+1}-t_i$, for all $i=0,1, \cdots, N-1$ and $h<1$ (small). We also assume that our processes are all one-dimensional, that is, $d=d_0=m=1$.
\\[4pt]
For $\theta_j:=(x_j,y_j,z_j),j=1,2,$ and for $\varphi=b,\sigma,f$, denote
\begin{align}
    g_1(x_1,x_2)&=1_{x_1\ne x_2} [g(V_T,x_1)-g(V_T,x_2)]/[x_1-x_2];\\
    \varphi_1(t,\theta_1,\theta_2)&=1_{x_1\ne x_2} [\varphi(t,V_t,x_1,y_1,z_1)-\varphi(t,V_t,x_2,y_1,z_1)]/[x_1-x_2];\\
    \varphi_2(t,\theta_1,\theta_2)&=1_{y_1\ne y_2} [\varphi(t,V_t,x_2,y_1,z_1)-\varphi(t,V_t,x_2,y_2,z_1)]/[y_1-y_2];\\
     \varphi_3(t,\theta_1,\theta_2)&=1_{z_1\ne z_2} [\varphi(t,V_t,x_2,y_2,z_1)-\varphi(t,V_t,x_2,y_2,z_2)]/[z_1-z_2].
\end{align}
Then we define 
\begin{align}
    \overline{F}(t,y)&=\esssup \left(\sup_{x_1\neq x_2,y_1\neq y_2,z_1\neq z_2}F(\theta_1,\theta_2;t,y)\right),\\
    \underline{F}(t,y)&=\essinf \left(\inf_{x_1\neq x_2,y_1\neq y_2,z_1\neq z_2}F(\theta_1,\theta_2;t,y)\right),
\end{align}
where
\begin{equation}
    F(\theta_{1},\theta_2;t,y)=f_1+f_2y+y(b_1+b_2y)+\frac{(f_3+b_3y)y(\sigma_1+\sigma_2y)}{1-\sigma_3y}.
\end{equation}

\begin{ass}\label{A1-wellp}
  There are three constants $c_1,c_2,c_3$ satisfying 
  \begin{equation}
      c_1>0,\hspace{4mm}0<c_2<c_3,\hspace{4mm} c_1c_3<1,
  \end{equation}
  and also there exists another constant $\epsilon=\epsilon(T)>0$ such that either one of the following three conditions holds:
  \begin{enumerate}[({Case} I)]
      \item $\abs{\sigma_3}\leq c_1,\abs{g_1}\leq c_2$; and $\overline{F}(t,c_3)\leq \epsilon,\underline{F}(t,-c_3)\geq -\epsilon$.
      \item $\abs{\sigma_3}\geq c_1^{-1},\abs{g_1}\geq c_2^{-1}$, and both $\sigma_3$ and $g_1$ keep the same sign, meaning, the dependence of the forward diffusion on the backward control variable $Z$ (through $\sigma_3$) and the slope of the terminal condition $g_1$ are both uniformly bounded away from zero, and they act in the same direction — either both positive or both negative.
      \item $\sigma_3g_1\leq c_1c_2$, and either $\sigma_3$ or $g_1$ keeps the same sign, meaning, the coefficients $\sigma_3$ and $g_1$ are not required to share the same sign, but at least one remains sign-consistent, and their combined effect is bounded so the system remains stable.
  \end{enumerate}
\end{ass}

\begin{ass}\label{A2-wellp}
    The coefficients $(b,\sigma,f)(\cdot, X,Y,Z)$ are $(\cF_t)_{t\in[0,T]}$-progressively measurable, for each $(X,Y,Z)\in \bR^d\times\bR^{d_0}\times\bR^{d_0\times m}$, and the terminal condition $g(\cdot,X)$ is $\cF_T$-measurable for any fixed $X\in\bR^d$. All coefficient functions may be random, but their randomness is through the dependence on the process $V$. Moreover, the following integrability condition holds:
     \begin{equation}\label{integ_initial_data}
         I^2_{0}:=\bE\bigg[\int_0^T\big(|b|^2+|f|^2+|\sigma|^2\big)(t,V_t,0,0,0)dt+|g(V_T,0)|^2\bigg].
     \end{equation}
\end{ass}

\begin{ass}\label{ass-1}
 There are positive constants $k_b,\, k_f,\,K,\, b_y,\, b_z,\, f_x,\, f_z,\,\sigma_x,\,\sigma_y,\,\sigma_z,$ and $g_x$ such that for all $(x_i,y_i,z_i)\in\bR^d\times\bR^{d_0}\times\bR^{d_0\times m} $, $i=1,2$, it holds that a.s.,
 \begin{enumerate}
 \item [(1)]
\begin{equation}
\left\{\begin{array}{l}\label{assumption1_1}
{[b(t,V_t,x_1,y_1,z_1)-b(t,V_t,x_2,y_1,z_1)]\cdot \triangle x\le k_b|\triangle x|^2},\\
{[f(t,V_t,x_1,y_1,z_1)-f(t,V_t,x_1,y_2,z_1)]\cdot \triangle y\le k_f|\triangle y|^2};
\end{array}\right.
\end{equation}
 \item[(2)]
\begin{equation}
\left\{\begin{array}{l}\label{assumption1_2}
{|b(t,V_t,x_1,y_1,z_1)-b(t,V_t,x_2,y_2,z_2)]|^2\le K|\triangle x|^2+b_y|\triangle y|^2+b_z|\triangle z|^2},\\
{|f(t,V_t,x_1,y_1,z_1)-f(t,V_t,x_2,y_2,z_2)|^2\le f_x|\triangle x|^2+K|\triangle y|^2+f_z|\triangle z|^2},\\
{|\sigma(t,V_t,x_1,y_1,z_1)-\sigma(t,V_t,x_2,y_2,z_2)|^2\le \sigma_x|\triangle x|^2+\sigma_y|\triangle y|^2+\sigma_z|\triangle z|^2},\\
{|g(V_T,x_1)-g(V_T,x_2)|^2\le g_x|\triangle x|^2},\\
\end{array}\right.
\end{equation}
where $(\triangle x, \triangle y, \triangle z)=(x_1-x_2,y_1-y_2,z_1-z_2)$.
\end{enumerate}
\end{ass}

\begin{ass}\label{ass-2}
There exists a continuous and increasing function $\rho:[0,\infty)\rightarrow[0,\infty)$ with $\rho(0)=0$ such that for any $0\leq t_i\leq t_j\leq T$, it holds that
\begin{equation}
    \bE\bigg[|\phi(t_i,V_{t_i},x,y,z)-\phi(t_j,V_{t_j},x,y,z)|^2\bigg]\leq \rho(|t_i-t_j|),
\end{equation}
where $\phi=b,\sigma,f$.
\end{ass}

Here, we present an example of fully coupled forward-backward stochastic differential equations of the form \eqref{maineq}, that satisfy the Assumptions \ref{A1-wellp}-\ref{ass-2}. 
\begin{exam}
    Fix $T>0$ and let $W$ be a 1-dimensional Brownian motion and $V_t$ be the following Ornstein–Uhlenbeck process: 
    \begin{equation}\label{Vt_fbm}
        V_t = v_0 + \widehat{W}_t, 
\end{equation}
where, $\widehat{W}_t$ is the fractional Brownian motion process with Hurst index $H\in(0,\frac{1}{2})$ and is defined as
\begin{equation}
	\widehat{W}_t \coloneqq \int_0^t \mathcal K(t-s) \, d W_s, \quad \mathcal K(r) \coloneqq
	\sqrt{2H} r^{H-1/2}, \quad r > 0.
\end{equation}
Then, with some bounded constants $s_3,k_1\neq 0$, we define the coupled forward-backward stochastic differential equations as:
\begin{equation}\label{example_fbsde}
    \begin{cases}
        dX_t &= \big(-\sin{X_t}+\tanh{Y_t}+Z_t+\cos{V_t}\big)dt + \big(X_t+Y_t+s_3 Z_t+V_t\big)dW_t, \ t\in[0,T],\\
        X_0&=0, \\
        -dY_t &=\big(X_t+\cos{Y_t}+\sin{Z_t}+\sqrt{|V_t|}\big)dt - Z_t dW_t, \ t\in[0,T],\\
        Y_T &= k_1 X_T +\int_0^T \sin{V_t}dt.
    \end{cases}
\end{equation}
\end{exam}

Notice that, here $\sigma_3=s_3\neq 0$ and $g_1=k_1\neq 0$ ensures that coefficient function $\sigma$ is monotone in $Z$ and $g$ is monotone in $X$. And we can find positive constants $c_,c_2,c_3$ in Assumption \ref{A1-wellp} based on $s_3$ and $k_1$. Also, coefficient functions $b,f,\sigma$ and $g$ satisfy corresponding conditions from Assumptions \ref{A2-wellp}-\ref{ass-2}.

%{\color{red}
% Give an example(s) satisfying all the assumptions above. It could be linear, (nontrivially) fully coupled.}

%**********************************

% In what follows, in contrast with \eqref{integ_initial_data}, we set 
% \begin{equation}
%         I^2_{0}=\bE\bigg[\int_0^T\big(|b|^2+|f|^2\big)(t,V_t,0,0,0)dt+|g(V_T,0)|^2 \bigg].
%     \end{equation} 
% \textcolor{orange}{Comment: $I_0$ is not needed and hence been removed and rest of the paper is adjusted accordingly.}\todo{Is this term $I_0$ needed?}

\subsection{Solutions to FBSDEs}
The following result about the well-posedness of non-Markovian coupled FBSDEs owes to Ma et al. \cite{JMaZWuDZhangJZhang2015OnWellPossedofFBSDEsAUnifiedAppr}.
\begin{thm}[Wellpossedness of FBSDE]\label{Wellp_thm}
Suppose all the processes involved are one-dimensional and that Assumptions \ref{A1-wellp}  -\ref{ass-1} hold. Then it holds that 
\begin{enumerate}[(1)]
	
	 \item FBSDE \eqref{maineq} admits a unique solution $(X,Y,Z)$, and there exists a constant $C>0$,depending only on $T$, the Lipschitz constants in Assumption \ref{ass-1}, and $c_1,c_2,c_3$, such that
	\begin{equation}
		\bE\bigg[\sup_{0\leq t\leq T} \bigg(|X_t|^2+|Y_t|^2\bigg)+\int_0^T |Z_t|^2 dt\bigg]\leq C\big(\bE|x_0|^2+I_{0}^2\big);
	\end{equation}
    \item FBSDE \eqref{maineq} possesses a decoupling field $u$ such that $Y_t=u(t,X_t)$ and  $u_1:=\frac{u(t,x_1)-u(t,x_2)}{x_1-x_2}$ satisfies the corresponding property of $g_1$ with $c_2$ being replaced by $c_3$ in Assumption \ref{A1-wellp}.

\end{enumerate}
\end{thm}

\begin{rmk}
    Here, we would like to clarify the property of the decoupling field $u$ in terms of the difference defined as $u_1$. Recall that, for the existence of the unique solution of FBSDE \eqref{maineq}, we need to satisfy only one of the three cases in Assumption \ref{A1-wellp}. And corresponding to each of these cases, the resulting decoupling field will satisfy the following:
    \begin{enumerate}[({Case} I)]
      \item $\abs{u_1}\leq c_3$.
      \item $\abs{u_1}\geq c_3^{-1}$, and $u_1$ keeps the same sign, meaning $u$ is monotone with respect to $X$.
      \item $\sigma_3u_1\leq c_1c_3$, and if $g$ is monotone with respect to $X$ then $u$ is also monotone with respect to $X$.
  \end{enumerate}
\end{rmk}

\noindent
Next, we state and prove the following continuity result.
\begin{thm}[Continuity of Solution]\label{stability_FBSDE}
 Under the assumptions of Theorem \ref{Wellp_thm} we have 
 \begin{align}
     &\sup_i \bE\bigg[ \sup_{t\in[t_i,t_{i+1}]} \bigg(|X_t-X_{t_i}|^2+|Y_t-Y_{t_i}|^2\bigg)\bigg]\nonumber\\
     &\leq C\big(I_{0}^2+\bE|x_0|^2\big)h+C\sup_i\bE\left[\int_{t_i}^{t_{i+1}}|Z_s|^2ds+\int_{t_i}^{t_{i+1}}|\sigma|^2(s,V_s,0,0,0)ds\right].
\end{align}

\end{thm}
\begin{proof}
Assume $0\leq v\leq t\leq T$. From FBSDE \eqref{maineq}, we have the following truncated form:
\begin{equation}\label{maineq_trunc}
    \begin{cases}
        X_t &= X_v+\int_v^t b(s,V_s,X_s,Y_s,Z_s)ds + \int_v^t \sigma(s,V_s,X_s,Y_s,Z_s)dW_s, \\
        Y_v &=Y_t+\int_v^t f(s,V_s,X_s,Y_s,Z_s)ds - \int_v^t Z_s dW_s.
    \end{cases}
\end{equation}
For the forward component, we estimate
\begin{equation*}
\begin{split}
    |X_t-X_v|^2 &= \bigg|\int_v^t b(s,V_s,X_s,Y_s,Z_s)ds\bigg|^2 + \bigg|\int_v^t \sigma(s,V_s,X_s,Y_s,Z_s)dW_s\bigg|^2\\
    &+2\bigg[\int_v^t b(s,V_s,X_s,Y_s,Z_s)ds \cdot\int_v^t \sigma(s,V_s,X_s,Y_s,Z_s)dW_s\bigg]
\end{split}
\end{equation*}
Applying Young's inequality gives
\begin{equation}\label{refpt_1}
    |X_t-X_v|^2 \leq 2\bigg|\int_v^t b(s,V_s,X_s,Y_s,Z_s)ds\bigg|^2 + 2\bigg|\int_v^t \sigma(s,V_s,X_s,Y_s,Z_s)dW_s\bigg|^2
\end{equation}
By the Cauchy-Schwarz inequality and the estimates from Theorem \ref{Wellp_thm}, we obtain
\begin{equation*}
    \begin{split}
        \bE\Bigg[\sup_{\tau \in[v, t]}\bigg|\int_v^{\tau} b(s,V_s,X_s,Y_s,Z_s)ds \bigg|^2\Bigg] &\leq \bE\bigg[\int_v^t 1^2 ds \cdot\int_v^t |b(s,V_s,X_s,Y_s,Z_s)|^2ds\bigg]\\
        &\leq (t-v)\bE\int_v^t2\bigg(|b(s,V_s,0,0,0)|^2+K|X_s|^2+b_y|Y_s|^2+b_z|Z_s|^2\bigg)ds\\
        &\leq C\big(I_{0}^2+\bE|x_0|^2\big)(t-v)+C(t-v)\bE\int_v^t|Z_s|^2ds.
    \end{split}
\end{equation*}
For the stochastic integral, Doob's inequality yields
\begin{align*}
        &\bE\Bigg[\sup_{\tau \in[v, t]}\bigg|\int_v^{\tau} \sigma(s,V_s,X_s,Y_s,Z_s)dW_s \bigg|^2 \Bigg]\\
        & \leq C\bE\int_v^t |\sigma(s,V_s,X_s,Y_s,Z_s)|^2ds\\
        &\leq C\bE\int_v^t\bigg(2|\sigma(s,V_s,0,0,0)|^2+\sigma_x|X_s|^2+\sigma_y|Y_s|^2+\sigma_z|Z_s|^2\bigg)ds\\
        %&\leq \int_v^t C\big(I_{0,\sigma}^2+\bE|x_0|^2\big)ds+C\bE\int_v^t \sigma_z|Z_s|^2ds\\
        &\leq  C\big(I_{0}^2+\bE|x_0|^2\big)(t-v)+C\bE\left[\int_{v}^{t}|Z_s|^2ds+\int_{v}^{t}|\sigma|^2(s,V_s,0,0,0)ds\right].
\end{align*}
 Taking the supremum  over $\tau\in [v,t]$ and then expectation on both sides of  \eqref{refpt_1}, and finally using the above estimates, we conclude that
\begin{equation*}
    \bE\Bigg[\sup_{\tau \in[v, t]} |X_{\tau}-X_v|^2\Bigg] \leq C\big(I_{0}^2+\bE|x_0|^2\big)(t-v)+C\bE\left[\int_{v}^{t}|Z_s|^2ds+\int_{v}^{t}|\sigma|^2(s,V_s,0,0,0)ds\right].
\end{equation*}
In particular, for any $i\in \{0,\cdots, N-1\}$,
\begin{equation*}
     \bE\bigg[\sup_{t\in[t_i,t_{i+1}]}|X_t-X_{t_i}|^2\bigg] \leq C\big(I_{0}^2+\bE|x_0|^2\big)h+C\bE\left[\int_{t_i}^{t_{i+1}}|Z_s|^2ds+\int_{t_i}^{t_{i
     +1}}|\sigma|^2(s,V_s,0,0,0)ds\right],
\end{equation*}
and hence,
\begin{equation}
\sup_{i}\bE\bigg[\sup_{t\in[t_i,t_{i+1}]}|X_t-X_{t_i}|^2\bigg] \leq C\big(I_{0}^2+\bE|x_0|^2\big)h+C\bE\sup_i\left[\int_{t_i}^{t_{i+1}}|Z_s|^2ds+\int_{t_i}^{t_{i+1}}|\sigma|^2(s,V_s,0,0,0)ds\right].
\end{equation}
\\[4pt]
For the backward equation in (\ref{maineq_trunc}), we similarly estimate
\begin{equation*}
\begin{split}
    |Y_t-Y_v|^2 &\leq \bigg|\int_v^t f(s,V_s,X_s,Y_s,Z_s)ds\bigg|^2 + \bigg|\int_v^t Z_sdW_s\bigg|^2\\
    &+2\bigg[\bigg|\int_v^t f(s,V_s,X_s,Y_s,Z_s)ds\bigg|.\bigg|\int_v^t Z_sdW_s\bigg|\bigg]\\
     &\leq 2\bigg|\int_v^t f(s,V_s,X_s,Y_s,Z_s)ds\bigg|^2 + 2\bigg|\int_v^t Z_sdW_s\bigg|^2.
    % &\leq C\big(I_0^2+\bE|x_0|^2\big)(t-v)+2\bE\int_{v}^{t}|Z_s|^2ds.
\end{split}
\end{equation*}
Using arguments similar to those for the forward component, we obtain
\begin{equation}
     \sup_{i}\bE\bigg[\sup_{t\in[t_i,t_{i+1}]}|Y_t-Y_{t_i}|^2\bigg] \leq C\big(I_{0}^2+\bE|x_0|^2\big)h+C\sup_i\bE\int_{t_i}^{t_{i+1}}|Z_s|^2ds.
\end{equation}
This completes the proof.
\end{proof}

\noindent
We now consider the Euler discretization scheme for the coupled FBSDE \eqref{maineq}:
\begin{align}\label{disc_fwbw}
	\begin{cases}
		X^\pi_{t_{i+1}} &= X^\pi_{t_i}+b(t_i,V_{t_i},X^\pi_{t_i},Y^\pi_{t_i},Z^\pi_{t_i})\Delta t_i + \sigma(t_i,V_{t_i},X^\pi_{t_i},Y^\pi_{t_i},Z^\pi_{t_i})\Delta W_{t_i},\\
		Y^\pi_{t_{i+1}} &=Y^\pi_{t_i}-f(t_i,V_{t_i},X^\pi_{t_i},Y^\pi_{t_i},Z^\pi_{t_i})\Delta t_i + Z^\pi_{t_i} \Delta W_{t_i},
	\end{cases}
\end{align}
where we use the forward representation for the backward SDE and $\Delta W_{t_i}:=W_{t_{i+1}}-W_{t_i}$. Taking conditional expectations with respect to $\cF_{t_i}$ on both sides of second equation in (\ref{disc_fwbw}), we obtain
\begin{equation*}
	Y^\pi_{t_i}=\bE\left[Y^\pi_{t_{i+1}}+f(t_i,v_{t_i},X^\pi_{t_i},Y^\pi_{t_i},Z^\pi_{t_i})\Delta t_i|\cF_{t_i}\right].
\end{equation*}
Multiplying $\Delta W_{t_i}$ on both sides of the second equation in (\ref{disc_fwbw}) and again taking conditional expectations with respect to $\cF_{t_i}$ again, we have:
\begin{equation*}
	Z_{t_i}^\pi=\frac{1}{\Delta t_i}\bE \left[Y^\pi_{t_{i+1}}\Delta W_{t_i}|\cF_{t_i}\right].
\end{equation*}
Inspired by the above observations and in the spirit of \cite{JZhang2004NumericalSchemeBSDEs}, we propose the following discrete approximation scheme:

\begin{equation}\label{discrete_eq}
	\left\{\begin{array}{l}
		\overline{X}_0^{\pi}=x_0,\\
		\overline{X}_{t_{i+1}}^{\pi}=\overline{X}_{t_{i}}^{\pi}+b\big(t_i,V_{t_i},\overline{X}_{t_{i}}^{\pi},\overline{Y}_{t_{i}}^{\pi},\overline{Z}_{t_{i}}^{\pi}\big)h+\sigma\big(t_i,V_{t_i},\overline{X}_{t_{i}}^{\pi},\overline{Y}_{t_{i}}^{\pi},\overline{Z}_{t_{i}}^{\pi}\big)\Delta W_{t_{i}},\\
		\overline{Y}_{t_N}^{\pi}=g(V_T,\overline{X}_{t_{N}}^{\pi}),\\
		\overline{Z}_{t_{i}}^{\pi}=\frac{1}{h}\bE\big[\overline{Y}_{t_{i+1}}^{\pi}\Delta W_{t_{i}}\big|\mathcal{F}_{t_i}\big],\\
		\overline{Y}_{t_{i}}^{\pi}=\bE\bigg[\overline{Y}_{t_{i+1}}^{\pi}+f\big(t_i,V_{t_i},\overline{X}_{t_{i}}^{\pi},\overline{Y}_{t_{i}}^{\pi},\overline{Z}_{t_{i}}^{\pi}\big)h\big|\mathcal{F}_{t_i}\bigg].
	\end{array}\right.
\end{equation}
\\[4pt]
Our next theorem in this section provides an estimate for the difference between the solution of the continuous FBSDE system (\ref{maineq}) and its discrete approximation (\ref{discrete_eq}). Before stating the theorem, we present two auxiliary lemmas whose proofs, along with the proof of the theorem, are given in the appendix.
\begin{lem}\label{R3_fw_est}
Fix $i$. For $l=1,2$, let
\begin{align*}
    X_{t_{i+1}}^l=X_{t_i}^l+b(t_i,V_{t_i},X_{t_i}^l,Y_{t_i}^l,Z_{t_i}^l)h+\int_{t_i}^{t_{i+1}}\alpha_t^l dt+\sigma(t_i,V_{t_i},X_{t_i}^l,Y_{t_i}^l,Z_{t_i}^l)\Delta W_{t_i}+\int_{t_i}^{t_{i+1}}\beta_t^ldW_t,
\end{align*}
where $X_{t_i}^l,Y_{t_i}^l,Z_{t_i}^l$ are $\mathcal{F}_{t_i}$-measurable and $\alpha^l,\beta^l$ are $(\mathcal{F}_t)_{t\in[0,T]}$-adapted. Then for any $\lambda_j>0,\;j=1,2$, it holds that
\begin{align}
    \bE_{t_i}\big[|\Delta X_{t_{i+1}}|^2\big]&\leq (1+A_1h)|\Delta X_{t_i}|^2+A_2h|\Delta Y_{t_i}|^2+A_3h|\Delta Z_{t_i}|^2\nonumber\\
    &+\bigg(1+\frac{4}{\lambda_1}+\frac{1}{\lambda_2}+\lambda_2\bigg)\bE_{t_i}\bigg[\int_{t_i}^{t_{i+1}}\big(|\Delta\alpha_t|^2+|\Delta\beta_t|^2\big)dt\bigg],
\end{align}
where the differences are defined as
\begin{equation*}
    \Delta X:=X^1-X^2,\hspace{2mm}\Delta Y:=Y^1-Y^2,\hspace{2mm}\Delta Z:=Z^1-Z^2,\hspace{2mm}\Delta \alpha:=\alpha^1-\alpha^2,\hspace{2mm}\Delta \beta:=\beta^1-\beta^2,
\end{equation*}
and the constants $A_1, A_2, A_3$ are given by
\begin{align*}
    A_1&=\lambda_1+\bigg(\frac{2}{\lambda_1}+h+\lambda_2h\bigg)K+(1+\lambda_1)\sigma_x,\\
    A_2&=\bigg(\frac{2}{\lambda_1}+h+\lambda_2h\bigg)b_y+(1+\lambda_1)\sigma_y,\\
    A_3&=\bigg(\frac{2}{\lambda_1}+h+\lambda_2h\bigg)b_z+(1+\lambda_1)\sigma_z.
\end{align*}
\end{lem}

\begin{lem}\label{R4_bw_est}
 Fix $i$. For $l=1,2$, let
\begin{align*}
    Y_{{t_i}}^l=Y_{t_{i+1}}^l+f(t_i,V_{t_i},X_{t_i}^l,Y_{t_i}^l,\hat{Z}_{t_i}^l)h+\int_{t_i}^{t_{i+1}}\gamma_t^l dt-\int_{t_i}^{t_{i+1}}Z_t^ldW_t,
\end{align*}
where $X_{t_i}^l,Y_{t_i}^l,Z_{t_i}^l$ are $\mathcal{F}_{t_i}$-measurable, $\gamma^l$ is $(\mathcal{F}_t)_{t\in[0,T]}$-adapted and 
\begin{equation*}
    \hat{Z}_{t_i}^l:=\frac{1}{h}\bE_{t_i}\big[Y_{t_{i+1}}^l\Delta W_{t_{i}}\big].
\end{equation*}
Then for any $\lambda_j>0,\;j=3,4,5$, it holds that
\begin{align}
    (1-A_6h)|\Delta Y_{{t_i}}|^2+A_7h|\Delta \hat{Z}_{{t_i}}|^2&\leq e^{A_4h}\bE_{t_i}|\Delta Y_{{t_{i+1}}}|^2+A_5h|\Delta X_{t_i}|^2\nonumber\\
    &+\left(1+\frac{1}{\lambda_4h}+\frac{1}{\lambda_3}\right)\bE_{t_i}\bigg|\int_{t_i}^{t_{i+1}}\Delta\gamma_tdt\bigg|^2,
\end{align}
where the differences are defined as
\begin{equation*}
    \Delta X:=X^1-X^2,\hspace{2mm}\Delta Y:=Y^1-Y^2,\hspace{2mm}\Delta \hat{Z}:=\hat{Z}^1-\hat{Z}^2,\hspace{2mm}\Delta \gamma:=\gamma^1-\gamma^2,
\end{equation*}
and the constants $A_4, A_5, A_6, A_7$ are given by
\begin{align*}
    A_4&=\frac{1}{h}\ln(1+\tilde{A}_4h),\\
    \tilde{A}_4&=\lambda_4+\frac{1}{\lambda_5}(1+\lambda_4h),\\
    A_5&=(1+\lambda_4h )f_x(h+\lambda_5),\\
    A_6&=(1+\lambda_4h)K(h+\lambda_5),\\
    A_7&=1-\tilde{A}_7,\\
    \tilde{A}_7&=\lambda_3+(1+\lambda_4h)(h+\lambda_5)f_z.
\end{align*}
\end{lem}

\begin{rmk}
It is worth noting that, for $i=1,2,\cdots,7$, the limit $\overline{A}_i=\lim_{h\rightarrow 0} A_i$ exists and is finite, which will be useful in analyzing the behaviour of the scheme as the time step $h$ tends to zero.
\end{rmk}

\begin{thm}\label{Rmain_convergence}
Suppose that $\big(\{\overline{X}^{\pi}_{t_i}\}_{0\le i\le N},\{\overline{Y}^{\pi}_{t_i}\}_{0\le i\le N},\{\overline{Z}^{\pi}_{t_i}\}_{0\le i\le N-1}\big)$ is a solution of (\ref{discrete_eq}) with $\overline{X}_{t_{i}}^{\pi},\overline{Y}_{t_{i}}^{\pi},\overline{Z}_{t_{i}}^{\pi}\in \mathbb{L}^2(\Omega,\mathcal{F}_{t_i},\bP)$. Also assume that the Assumption \ref{ass-2} and all the assumptions of Theorem \ref{Wellp_thm} hold true and there exist constants $\lambda_i>0,i=1,2,3,4,5$ such that $A_6h<1,\tilde{A}_7<1$ and 
\begin{equation}
    \bigg(\frac{\overline{A}_2(1-e^{-(\overline{A}_1+\overline{A}_4)T})}{\overline{A}_1+\overline{A}_4}+\frac{\overline{A}_3}{\overline{A}_7}\bigg) \bigg(e^{(\overline{A}_1+\overline{A}_4+\overline{A}_6)T}g_x+ \frac{\overline{A}_5(e^{(\overline{A}_1+\overline{A}_4+\overline{A}_6)T}-1)}{\overline{A}_1+\overline{A}_4}\bigg)<1.
\end{equation}
Then we have
\begin{equation}
\begin{split}
     \sup_{t\in[0,T]} \bE\bigg[|X_t-\overline{X}_{t}^{\pi}|^2+|Y_t-\overline{Y}_{t}^{\pi}|^2\bigg]+\bE\int_{0}^{T}|\hat{Z}_t-\overline{Z}_{t}^{\pi}|^2dt&\leq CM,
\end{split}
\end{equation}
where
\begin{align*}
   M=&\big(I_{0}^2+\bE|x_0|^2\big)h+\rho(h)+\frac{1}{h}\sup_{0\le k\le N-1}\bE\int_{t_k}^{t_{k+1}}|Z_t-\tilde{Z}_{t_k}|^2dt\\
   &+\sup_{0\leq k\leq N-1}\bE\left[\int_{t_k}^{t_{k+1}}|Z_t|^2dt+\int_{t_k}^{t_{k+1}}\big(|f|^2+|\sigma|^2\big)(t,V_t,0,0,0)dt\right],
\end{align*}
and $\overline{X}_{t}^{\pi}=\overline{X}_{t_{i}}^{\pi},\overline{Y}_{t}^{\pi}=\overline{Y}_{t_i}^{\pi},\overline{Z}_{t}^{\pi}=\overline{Z}_{t_i}^{\pi},\hat{Z}_t=\hat{Z}_{t_i}$, for $t\in[t_i,t_{i+1})$ and 
    $\tilde{Z}_{t_i}:=\frac{1}{h}\bE\bigg[\int_{t_i}^{t_{i+1}}Z_tdt\big|\mathcal{F}_{t_i}\bigg]$, $\hat{Z}_{t_i}=\frac{1}{h}\bE\bigg[ Y_{t_{i+1}}\Delta W_{t_i}\big|\cF_{t_i}\bigg]$.
\end{thm}

\subsection{Neural Network Approximation}

Deep learning offers a powerful framework for the approximation of high-dimensional functions. In this section, we describe the architecture of fully connected feedforward neural networks and recall the universal approximation theorem that underpins their theoretical capacity.

Consider a fully-connected feedforward neural network with input dimension of $d^i$ and output dimension $d^o$, consisting of $J+1$ layers, where $J+1\in \bN\setminus\{1,2\}$. Let the number of neurons in each layer be denoted by $k_n$ for $ n=0,\ldots,J$, where $k_0=d^i \ \text{and} \ k_N=d^o$. For simplicity, assume that all hidden layers contain the same number of neurons, i.e., $k_n=k$ for all $n=1,\cdots,J-1$. 

The neural network defines a function from $\bR^{d^i}$ to $\bR^{d^o}$, represented as a composition of affine transformations and activation functions:

\begin{equation}\label{nnfunc}
    x\in\bR^{d^i} \mapsto B_J \circ \varrho \circ B_{J-1} \circ \cdots \circ \varrho \circ B_1(x)\in \bR^{d^o}.
\end{equation}
Here, the composition is defined as $f_1\circ f_2(x)=f_1(f_2(x))$. Each $B_n$ denotes an affine transformation corresponding to layer $n$, defined as
\[ B_n(x)=\cW_nx+\beta_n, \]
where $\cW_n$ is the weight matrix and $\beta_n$ is the bias vector of the $n$th layer of the network. Specifically, $B_1:\bR^{d^i}\mapsto \bR^{k}, B_J:\bR^k\mapsto\bR^{d^o}, \ \text{and} \ B_n: \bR^k \mapsto \bR^k, n=2,\ldots,J-1$. The activation function $\varrho$ is applied component-wise to the output of each hidden layer. The final layer does not include a nonlinearity and is taken to be the identity function.  

Let $\theta=(\mathcal{W}_n,\beta_n)_{n=1}^J$ denote the collection of all trainable parameters in the network. For a given architecture specified by $(d^i,d^o,J,k$, the total number of parameters in a network is given by 
$$J_{k}=\sum_{n=0}^{J-1}(k_n+1)k_{n+1}=(d^i+1)k+(k+1)k(J-2)+(k+1)d^o,$$ so that $\theta\in\bR^{J_k}$. Let $\Theta$ be the set of all possible values of $\theta$, and if there are no constraints on parameters, then $\Theta=\bR^{J_k}$. We denote the neural network function defined in (\ref{nnfunc}) by $\mathcal{X}^{\mathcal{N}}(\cdot;\theta)$.   The class of all such neural networks, with a fixed structure $(d^i,d^o,J,k)$ and activation function $\varrho$ is denoted by $\mathcal{NN}_{d^i,d^o,J,k}^{\varrho}(\Theta)$.

Here, we state the generalized universal approximation theorem by Bayer et al. \cite{CBayerJQiuYYao2020PricingOptnRoughVolatBSPDEs}, who extend the universal approximation theorem for the approximation of functions defined on a finite-dimensional space to random functions in an infinite-dimensional space.
\begin{lem}\label{lem-approximation}
    For each $T_0\in(0,T],J\in\bN^+\setminus\{1\},$ and $d^i,d^o\in\bN^+$, the function set 
    \begin{align*}
        \big\{\Phi_k(W_{t_1},\cdots,W_{t_i},x;\theta):\Phi_k(\cdot;\theta)\in \cN\cN^\varrho_{d^i+i,d^o,J,k}(\bR^{J_k}),k,i\in\bN^+,0<t_1<t_2\cdots<t_i\leq T_0\big\}
    \end{align*}
    is dense in $L^2(\Omega\times \bR^{d^i},\cF_{T_0}\otimes\cB(\bR^{d^i}),\bP(\omega)\otimes dx;\bR^{d^o})$, whenever $\varrho$ is continuous and non-constant.
\end{lem}
The above universal approximation theorem gives the justification to use neural network approximation to approximate the decoupling random field that appears in non-Markovian coupled FBSDEs.
 As remarked in \cite{CBayerJQiuYYao2020PricingOptnRoughVolatBSPDEs}, the $\sigma$-algebra $\mathcal F_{T_0}$ is generated by $\{W_s:s\in[0,T_0]\}$ where the process $W$ and accordingly the dependence of it in the assertion of Lemma \ref{lem-approximation} may be replaced by any continuous process.

\section{Numerical Scheme and Simulation Error for Non-Markovian Coupled-FBSDE}

 Without loss of generality, we take $d_0=1$ in what follows. Inspired by the approximation results in Lemma \ref{lem-approximation} under the non-Markovian framework, the backward processes $(Y_t,Z_t)$ is to be expressed  a function of the forward process $X_t$ and the historical paths $\{W_s,V_s:0\le s\le t\}$ (or just $\{W_s:0\le s\le t\}$ or $\{V_s:0\le s\le t\}$) because the randomness of the coefficients is from the continuous process $\{V_s:0\le s\le t\}$ or essentially from the Wiener process $\{W_s:0\le s\le t\}$. 
 
 Accordingly, we reformulate the problem of solving the FBSDE (\ref{maineq}) as the following stochastic optimization problem:
 \begin{align}\label{scheme_objective}
 	\inf_{\cY_0,\cZ_i}\quad \widehat{\mathfrak{L}}(\cY_0,\cZ_0,\cdots,\cZ_{N-1}):=\bE|g(V_T,X^\pi_{T})-Y^\pi_T|^2,
 \end{align}
 subject to,
 \begin{equation}\label{scheme_system}
 	\begin{cases}
 		X^{\pi}_0=x_0,\quad Y^\pi_0=\cY_0(x_0),\\
 		Z^{\pi}_{t_i}=\cZ_i\big(X^\pi_{t_i},\{W_s,V_s:0\le s\le t_i\}\big),\\
 		X^{\pi}_{t_{i+1}}=X^{\pi}_{t_i}+b(t_i,V_{t_i},X^{\pi}_{t_i},Y^{\pi}_{t_i},Z^{\pi}_{t_i})\triangle t_i+\sigma(t_i,V_{t_i},X^{\pi}_{t_i},Y^{\pi}_{t_i},Z^{\pi}_{t_i})\triangle W_{t_i},\\
 		Y^{\pi}_{t_{i+1}}=Y^\pi_{t_i}-f(t_i,V_{t_i},X^{\pi}_{t_i},Y^{\pi}_{t_i},Z^{\pi}_{t_i})\Delta t_i+Z^\pi_{t_i} \Delta W_{t_i}.
 	\end{cases}
 \end{equation}
 Our next result, Theorem \ref{Error_Estimate_Num_Scheme}  will show that the error between the solution $(X_t,Y_t,Z_t)$, of the continuous FBSDE (\ref{maineq}) and its discrete counterpart $(X^{\pi}_t,Y^{\pi}_t,Z^{\pi}_t)$, defined in (\ref{scheme_system}), is bounded by the objective function $\bE|g(V_T,X^\pi_{T})-Y^\pi_T|^2$ of the above optimization problem. 
 
 Therefore, if the functions $\cY_0,\cZ_0,\cdots,\cZ_{N-1}$ can be approximated using neural networks such that the objective function is minimized, ideally to zero, of the above optimization problem can be minimized to zero, then the resulting solution of the optimization problem will serve as an approximation to the solution of the original FBSDE (\ref{maineq}). Consequently, the formulation of the stochastic optimization problem in \eqref{scheme_objective}–\eqref{scheme_system}, along with the error estimates provided in the forthcoming theorem, forms the backbone of our proposed numerical scheme for solving non-Markovian, fully coupled-FBSDEs. A detailed implementation of this scheme will be presented in Section \ref{algorithm_Section}. 
 
 Before we proceed to Theorem~\ref{Error_Estimate_Num_Scheme}, we must introduce and prove a key Lemma.
 
 Inspired by the structure of the discretized system (\ref{discrete_eq}), we consider the following system of equations:
\begin{equation}\label{lemma1_system}
	\left\{\begin{array}{l}
		{X^{\pi}_0=x_0} ,\\
		{X^{\pi}_{t_{i+1}}=X^{\pi}_{t_i}+b(t_i,V_{t_i},X^{\pi}_{t_i},Y^{\pi}_{t_i},Z^{\pi}_{t_i})\triangle t_i+\sigma(t_i,V_{t_i},X^{\pi}_{t_i},Y^{\pi}_{t_i},Z^{\pi}_{t_i})\triangle W_{t_i}},\\
		{Z^{\pi}_{t_i}=\frac{1}{h}\bE[Y^{\pi}_{t_{i+1}} \triangle W_{t_i}|\mathcal{F}_{t_i}]},\\
		{Y^{\pi}_{t_i}=\bE[Y^{\pi}_{t_{i+1}}+f(t_i,V_{t_i},X^{\pi}_{t_i},Y^{\pi}_{t_i},Z^{\pi}_{t_i})h|\mathcal{F}_{t_i}]}.
	\end{array}\right.
\end{equation}
Note that the key distinction between (\ref{discrete_eq}) and (\ref{lemma1_system}) lies in the absence of a specified terminal condition in (\ref{lemma1_system}), resulting in the system admitting infinitely many solutions. In particular, both $(X^{\pi},Y^{\pi},Z^{\pi})$, as defined in (\ref{scheme_system}) and $(\overline X^{\pi},\overline Y^{\pi},\overline Z^{\pi})$ from (\ref{discrete_eq}), satisfy the system (\ref{lemma1_system}). 
\begin{lem}\label{Lemma1}
Let $j=1,2$, and suppose that $\big(\{X^{\pi,j}_{t_i}\}_{0\le i\le N},\{Y^{\pi,j}_{t_i}\}_{0\le i\le N},\{Z^{\pi,j}_{t_i}\}_{0\le i\le N-1}\big)$ are two solutions of the system (\ref{lemma1_system}) with $X^{\pi,j}_{t_{i+1}},Y^{\pi,j}_{t_{i+1}}\in \mathbb{L}^2(\Omega,\mathcal{F}_{t_{i+1}},\mathbb{P})$ for all $0\le i\le N-1$, $j=1,2$. Let $\lambda_0,\lambda_1>0, \lambda_2> f_z$, and suppose $h$ is sufficiently small. Define the constants:
\begin{align}
A_1&:=Kh+\sigma_x+2k_b+\lambda_0+\lambda_1,\nonumber\\
A_2&:=(h+\lambda_1^{-1})b_y+\sigma_y,\nonumber\\
A_3&:=(h+\lambda_0^{-1})b_z+\sigma_z,\nonumber\\
A_4&:=-h^{-1}\ln[1-(2k_f+\lambda_2)h],\nonumber\\
A_5&:=f_x\lambda_2^{-1}\big(1-(2k_f+\lambda_2)h\big)^{-1},\nonumber\\
A_6&:=(1-f_z\lambda_2^{-1})\big(1-(2k_f+\lambda_2)h\big)^{-1}.\nonumber
\end{align}
Let $\delta X_{i}=X_{t_{i}}^{\pi, 1}-X_{t_{i}}^{\pi, 2}, \delta Y_{i}=Y_{t_{i}}^{\pi, 1}-Y_{t_{i}}^{\pi, 2}$. Then for any $1 \leq n \leq N$, the following estimates hold:
\begin{align*}
\bE|\delta X_{n}|^{2} \leq A_{2} \sum_{i=0}^{n-1} e^{A_{1}(n-i-1) h} \bE|\delta Y_{i}|^{2}h+A_3 \sum_{i=0}^{n-1} e^{A_{1}(n-i-1) h} \bE|\delta Z_{i}|^{2}h,
\end{align*}
and
\begin{align}
\bE|\delta Y_{n}|^{2}+\sum_{i=n}^{N-1}A_6e^{A_4(i-n) h}\bE|\delta Z_i|^2h \leq e^{A_4(N-n) h} \bE|\delta Y_{N}|^{2}+A_5 \sum_{i=n}^{N-1} e^{A_4(i-n) h} \bE|\delta X_{i}|^{2} h.\nonumber
\end{align}
\end{lem}

\begin{proof}
Define the following differences:
\begin{equation}
\left\{\begin{array}{l}
{\delta Z_i=Z^{\pi,1}_{t_i}-Z^{\pi,2}_{t_i}},\\
{\delta b_i=b(t_i,V_{t_i},X^{\pi,1}_{t_i},Y^{\pi,1}_{t_i},Z^{\pi,1}_{t_i})-b(t_i,V_{t_i},X^{\pi,2}_{t_i},Y^{\pi,2}_{t_i},Z^{\pi,2}_{t_i})},\\
{\delta\sigma_i=\sigma(t_i,V_{t_i},X^{\pi,1}_{t_i},Y^{\pi,1}_{t_i},Z^{\pi,1}_{t_i})-\sigma(t_i,V_{t_i},X^{\pi,2}_{t_i},Y^{\pi,2}_{t_i},Z^{\pi,2}_{t_i})},\\
{\delta f_i=f(t_i,V_{t_i},X^{\pi,1}_{t_i},Y^{\pi,1}_{t_i},Z^{\pi,1}_{t_i})-f(t_i,V_{t_i},X^{\pi,2}_{t_i},Y^{\pi,2}_{t_i},Z^{\pi,2}_{t_i})}.
\end{array}\right.
\end{equation}
From system \eqref{lemma1_system}, we then obtain the update rules:
\begin{equation}
\left\{\begin{array}{l}\label{delta_XYZ}
{\delta X_{i+1}=\delta X_{t_i}+\delta b_i h+\delta\sigma_i\triangle W_{t_i}},\\
{\delta Z_i=\frac{1}{h}\bE[\delta Y_{i+1} \triangle W_{t_i}|\mathcal{F}_{t_i}]},\\
{\delta Y_i=\bE[\delta Y_{i+1}+\delta f_i h|\mathcal{F}_{t_i}]}.
\end{array}\right.
\end{equation}
\textbf{Step 1: Estimate for $\delta X$}\\
Squaring and then taking expectations on both sides of the first equation of \eqref{delta_XYZ}, we get
\begin{align}\label{X_est}
	\bE|\delta X_{i+1}|^2=\bE|\delta X_i+\delta b_i h|^2+h\bE|\delta\sigma_i|^2.
\end{align}

Using Assumptions \ref{ass-1} and \ref{ass-2}, and applying standard inequalities, we find that for any $\lambda_0,\lambda_1>0$,
\begin{align}
&\bE|\delta X_{i+1}|^2
-\bE|\delta X_i|^2 
\nonumber\\
&=\bE|\delta b_i|^2 h^2 + h\bE|\delta\sigma_i|^2+2h\bE\big[\delta X_i \cdot (b(t_i,V_{t_i},X^{\pi,1}_{t_i},Y^{\pi,1}_{t_i},Z^{\pi,1}_{t_i})-b(t_i,V_{t_i},X^{\pi,2}_{t_i},Y^{\pi,2}_{t_i},Z^{\pi,2}_{t_i}))\big]\nonumber\\
&\le 
\big(K\bE|\delta X_i|^2+b_y\bE|\delta Y_i|^2+b_z\bE|\delta Z_i|^2\big)h^2 + \big(\sigma_x\bE|\delta X_i|^2+\sigma_y\bE|\delta Y_i|^2+\sigma_z\bE|\delta Z_i|^2\big)h\nonumber\\
&\;+2hk_b\bE|\delta X_i|^2+2h\bE\bigg[\sqrt{b_y|\delta Y_i|^2+b_z|\delta Z_i|^2}\,\, |\delta X_i |\bigg]\nonumber\\
&\le 
 \big(K\bE|\delta X_i|^2+b_y\bE|\delta Y_i|^2+b_z\bE|\delta Z_i|^2\big)h^2 + \big(\sigma_x\bE|\delta X_i|^2+\sigma_y\bE|\delta Y_i|^2+\sigma_zE|\delta Z_i|^2\big)h\nonumber\\
&\;+2hk_b\bE|\delta X_i|^2+2h\bE\bigg[\sqrt{b_y|\delta Y_i|^2|\delta X_i|^2}+\sqrt{b_z|\delta Z_i|^2|\delta X_i|^2}\bigg]\nonumber\\
&\le 
\big(K\bE|\delta X_i|^2+b_y\bE|\delta Y_i|^2+b_z\bE|\delta Z_i|^2\big)h^2 + \big(\sigma_x\bE|\delta X_i|^2+\sigma_y\bE|\delta Y_i|^2+\sigma_z\bE|\delta Z_i|^2\big)h\nonumber\\
&\;+2hk_b\bE|\delta X_i|^2+h\lambda_1\bE|\delta X_i|^2+h\lambda_1^{-1}b_yE|\delta Y_i|^2+h\lambda_0\bE|\delta X_i|^2+h\lambda_0^{-1}b_zE|\delta Z_i|^2\nonumber\\
&=(Kh+\sigma_x+2k_b+\lambda_1+\lambda_0)h \, \bE|\delta X_i|^2+\big((h+\lambda_1^{-1})b_y+\sigma_y\big)h\bE|\delta Y_i|^2\nonumber\\
&\;+\big((h+\lambda_0^{-1})b_z+\sigma_z\big)h\bE|\delta Z_i|^2.
\end{align}

Denoting
\begin{align*}
    A_1&=Kh+\sigma_x+2k_b+\lambda_1+\lambda_0,\\
    A_2&=(h+\lambda_1^{-1})b_y+\sigma_y,\\
    A_3&=(h+\lambda_0^{-1})b_z+\sigma_z,
\end{align*}
we obtain the recurrence:
\begin{equation*}
    \bE|\delta X_{i+1}|^2 \leq (1+A_1 h)\bE |\delta X_i|^2+A_2h\bE|\delta Y_i|^2+A_3h\bE|\delta Z_i|^2
\end{equation*}

 Then using the fact $\bE|\delta X_0|^2=0$, we obtain by induction that, for $1\le n\le N$,
\begin{align}
\bE|\delta X_n|^2&\le(1+A_1h)^n\bE|\delta X_0|^2+\sum_{i=0}^{n-1}(1+A_1h)^{n-i-1}A_2h\bE|\delta Y_i|^2+\sum_{i=0}^{n-1}(1+A_1h)^{n-i-1}A_3h\bE|\delta Z_i|^2\nonumber\\
&\le\sum_{i=0}^{n-1}e^{A_1h(n-i-1)}A_2h\bE|\delta Y_i|^2+\sum_{i=0}^{n-1}e^{A_1h(n-i-1)}A_3h\bE|\delta Z_i|^2.
\end{align}
\textbf{Step 2: Estimate for $\delta Y$ and $\delta Z$}\\
By the martingale representation theorem there exists an $\left(\mathcal{F}_t\right)_{t\geq 0}-$adapted square-integrable process $\{\delta \bar{Z}_t\}_{t_i\le t\le t_{i+1}}$ such that
\begin{equation}\label{delta_Y}
\delta Y_{i+1}=E[\delta Y_{i+1}|\mathcal{F}_{t_i}]+\int_{t_i}^{t_{i+1}}\delta \bar{Z}_t dW_t=\delta Y_i-\delta f_i h+\int_{t_i}^{t_{i+1}}\delta \bar{Z}_t dW_t.
\end{equation}
Squaring on both sides and taking the expectation yields that
\begin{align}\label{Y_est}
\bE|\delta Y_{i+1}|^2&=\bE|\delta Y_i-\delta f_i h|^2+\int_{t_i}^{t_{i+1}}\bE|\delta \bar{Z}_t|^2dt.
\end{align}

Rewriting this, and estimating the term $\delta f_i$, we obtain for any $\lambda_2>0$:
\begin{align}\label{EY}
&\bE|\delta Y_{i+1}|^2-\bE|\delta Y_i|^2\nonumber\\
&=\bE|\delta Y_i-\delta f_ih|^2-\bE|\delta Y_i|^2+\int_{t_i}^{t_{i+1}}\bE|\delta \bar{Z}_t|^2dt\nonumber\\
&\ge  \int_{t_i}^{t_{i+1}}\bE|\delta \bar{Z}_t|^2dt-2h\bE\Big[(f(t_i,V_{t_i},X^{1,\pi}_{t_i},Y^{1,\pi}_{t_i},Z^{1,\pi}_{t_i})-f(t_i,V_{t_i},X^{1,\pi}_{t_i},Y^{2,\pi}_{t_i},Z^{1,\pi}_{t_i}))\cdot\delta Y_i\Big]\nonumber\\
&\quad
-2h\bE\Big[(f(t_i,V_{t_i},X^{1,\pi}_{t_i},Y^{2,\pi}_{t_i},Z^{1,\pi}_{t_i})-f(t_i,V_{t_i},X^{2,\pi}_{t_i},Y^{2,\pi}_{t_i},Z^{2,\pi}_{t_i}))\cdot \delta Y_i\Big]\nonumber\\
&\ge \int_{t_i}^{t_{i+1}}\bE|\delta \bar{Z}_t|^2dt-2k_fh\bE|\delta Y_i|^2-\Big(\lambda_2\bE|\delta Y_i|^2+\lambda_2^{-1}(f_x\bE|\delta X_i|^2+f_z\bE|\delta Z_i|^2)\Big)\, h.
\end{align}
Using \eqref{delta_Y} into \eqref{delta_XYZ} we can have
$$\delta Z_i=\frac{1}{h}\bE\left[\int_{t_i}^{t_{i+1}}\delta \bar{Z}_tdt\Big|\mathcal{F}_{t_i}\right],$$
which further implies that
\begin{align}
\bE|\delta Z_i|^2 h= \frac{1}{h}\bE\Bigg|\bE\bigg[\int_{t_i}^{t_{i+1}}\delta \bar{Z}_t dt|\mathcal{F}_{t_i}\bigg]\Bigg|^2 \le \int_{t_i}^{t_{i+1}}\bE| \delta \bar{Z}_t |^2dt.
\end{align}
Plugging this into \eqref{EY} gives
$$
\bE|\delta Y_{i+1}|^2\ge \big(1-(2k_f+\lambda_2)h\big)\bE|\delta Y_i|^2+(1-f_z\lambda_2^{-1})h\bE|\delta Z_i|^2-f_x\lambda_2^{-1}h\bE|\delta X_i|^2.
$$
Then for any $\lambda_2> f_z$ and sufficiently small $h$ satisfying $(2k_f+\lambda_2)h<1$, we obtain
\begin{align*}
&\bE|\delta Y_i|^2+(1-f_z\lambda_2^{-1})\big(1-(2k_f+\lambda_2)h\big)^{-1}\bE|\delta Z_i|^2h
\\
&\le \big(1-(2k_f+\lambda_2)h\big)^{-1}\left(\bE|\delta Y_{i+1}|^2+f_x\lambda_2^{-1}\bE|\delta X_i|^2h\right).
\end{align*}
Denoting 
\begin{align*}
    A_4&=-h^{-1}ln[1-(2k_f+\lambda_2)h],\\
    A_5&=f_x\lambda_2^{-1}\big(1-(2k_f+\lambda_2)h\big)^{-1},\\
    A_6&=(1-f_z\lambda_2^{-1})\big(1-(2k_f+\lambda_2)h\big)^{-1},
\end{align*}
 we obtain by induction that, for $1\le n\le N$,
$$
\bE|\delta Y_n|^2+\sum_{i=n}^{N-1}A_6 e^{A_4(i-n) h} \bE|\delta Z_i|^2 h\le e^{A_4h(N-n)}\bE|\delta Y_N|^2+A_5\sum_{i=n}^{N}e^{A_4(i-n)h}\bE|\delta X_i|^2h.
$$
\end{proof}

\begin{rmk}\label{rmk:notation_lemma1}
Note that, using the notations from Lemma \ref{Lemma1}, for $i=1,2,\cdots,6$, the limits $\overline{A}_i=\lim_{h\rightarrow 0} A_i$ exists and are finite.
\end{rmk}
%\hspace{5mm}
\begin{thm}\label{Error_Estimate_Num_Scheme}
Assume the assumptions of Theorem \ref{Rmain_convergence} hold, and there exists constants $\lambda_i>0,i=0,1,2$ such that $\lambda_2>f_z$, $(2k_f+\lambda_2)h<1$ and $\overline{A}_0<1$ where 

$$\overline{A}_0=\left( \overline{A}_2\frac{1-e^{-(\overline{A}_1+\overline{A}_4)T}}{\overline{A}_1+\overline{A}_4}+\frac{\overline{A}_3}{\overline{A}_6}\right)\cdot
\left(g_x(1+\lambda_3)e^{(\overline{A}_1+\overline{A}_4)T}+\overline{A}_5\frac{e^{(\overline{A}_1+\overline{A}_4)T}-1}{\overline{A}_1+\overline{A}_4}\right),
$$
and the constants $\overline{A}_1,\cdots,\overline{A}_6$ are as defined in Lemma \ref{Lemma1} and Remark \ref{rmk:notation_lemma1}.

Then there exists a constant $C>0$, independent of h, d, and m, such that for sufficiently small h,
$$
\sup_{t \in[0, T]}\big(\bE|X_{t}-\hat{X}_{t}^{\pi}|^{2}+\bE|Y_{t}-\hat{Y}_{t}^{\pi}|^{2}\big)+\int_{0}^{T} \bE|Z_{t}-\hat{Z}_{t}^{\pi}|^{2}dt \leq C\big(M+\bE|g(V_T,X_{T}^{\pi})-Y_{T}^{\pi}|^{2}\big),
$$
where,
\begin{align*}
	M=&\big(I_{0}^2+\bE|x_0|^2\big)h+\rho(h)+\frac{1}{h}\sup_{0\le k\le N-1}\bE\int_{t_k}^{t_{k+1}}|Z_t-\tilde{Z}_{t_k}|^2dt\\
    &+\sup_{0\leq k\leq N-1}\bE\bigg[\int_{t_k}^{t_{k+1}}|Z_t|^2dt+\int_{t_k}^{t_{k+1}}\big(|f|^2+|\sigma|^2\big)(t,V_t,0,0,0)dt\bigg],
\end{align*} 
and piecewise constant approximations $\hat{X}_{t}^{\pi}, \hat{Y}_{t}^{\pi},\hat{Z}_{t}^{\pi} $ are defined as 
\begin{equation*}
    \hat{X}_{t}^{\pi}=X_{t_{i}}^{\pi},\; \hat{Y}_{t}^{\pi}=Y_{t_{i}}^{\pi}, \;\hat{Z}_{t}^{\pi}=Z_{t_i}^{\pi} \quad\text{for} \quad t \in[t_{i}, t_{i+1}),
\end{equation*}
 with $(X^{\pi},Y^{\pi},Z^{\pi})$ satisfying the discrete-time system (\ref{scheme_system}) and $\tilde{Z}_{t_k}:=\frac{1}{h}\bE\left[\int_{t_k}^{t_{k+1}}Z_tdt|\cF_{t_k}\right]$.
\end{thm}

\begin{proof}
We use the same notation as Lemma \ref{Lemma1}. Let $$X^{\pi,1}_{t_i}=X^{\pi}_{t_i},\;Y^{\pi,1}_{t_i}=Y^{\pi}_{t_i},\;Z^{\pi,1}_{t_i}=Z^{\pi}_{t_i},$$ and $$X^{\pi,2}_{t_i}=\overline{X}^{\pi}_{t_i},\;Y^{\pi,2}_{t_i}=\overline{Y}^{\pi}_{t_i},\;Z^{\pi,2}_{t_i}=\overline{Z}^{\pi}_{t_i},$$ 
  where $(X^{\pi},Y^{\pi},Z^{\pi})$ is defined in (\ref{scheme_system}) and $(\overline X^{\pi},\overline Y^{\pi},\overline Z^{\pi})$ is defined in (\ref{discrete_eq}). Both $(X^{\pi},Y^{\pi},Z^{\pi})$ and $(\overline X^{\pi},\overline Y^{\pi},\overline Z^{\pi})$ satisfy the system (\ref{lemma1_system}).
  
By the terminal condition, we have for any
$\lambda_3>0$,
\begin{equation}\label{est:delY_N}
	\bE|\delta Y_N|^2=\bE|g(V_T,\overline{X}^{\pi}_T)-Y^{\pi}_T|^2\le(1+\lambda_3^{-1})\bE|g(V_T,X^{\pi}_T)-Y^{\pi}_T|^2+g_x(1+\lambda_3)\bE|\delta X_N|^2.
\end{equation}
From Lemma \ref{Lemma1}, we also have the estimates:
\begin{align}\label{est:t321}
e^{-A_1nh}\bE|\delta X_n|^2\le&\sum_{i=0}^{n-1}e^{-A_1h(i+1)}A_2h\bE|\delta Y_i|^2+\sum_{i=0}^{n-1}e^{-A_1h(i+1)}A_3h\bE|\delta Z_i|^2,
%\le&\bigg(\sum_{i=0}^{n-1}e^{-A_1h(i+1)-A_4ih}A_2h+\frac{A_3}{A_6}e^{-A_1h}\bigg)S,
\end{align}
and
\begin{align}\label{est:t322}
	e^{A_4nh}\bE|\delta Y_n|^2
	+\sum_{i=n}^{N-1}e^{A_4ih}A_6\bE|\delta Z_i|^2h\le&e^{A_4T}\bE|\delta Y_N|^2+A_5\sum_{i=n}^{N-1}e^{A_4ih}\bE|\delta X_i|^2h,
	%\nonumber\\
	%\le&e^{A_4T}(1+\lambda_3^{-1})\bE|g(V_T,X^{\pi}_T)-Y^{\pi}_T|^2+\bigg(g_x(1+\lambda_3)e^{(A_1+A_4)T}
	%\nonumber\\
	% &+A_5\sum_{i=n}^{N-1}e^{(A_1+A_4)ih}h\bigg)P,
\end{align}
Now, we define:
\begin{align}
	P&:=\max_{0\le n\le N}e^{-A_1nh}\bE|\delta X_n|^2,
\end{align}
and 
\begin{align}
		S&:=\max\left\{\max_{0\le n\le N}(e^{A_4nh}\bE|\delta Y_n|^2),\,\,
	\sum_{i=0}^{N-1}e^{A_4ih}A_6h\bE|\delta Z_i|^2\right\}.
\end{align}
Then by combining \eqref{est:t321} and \eqref{est:t322}, and using estimate from \eqref{est:delY_N} we get
\begin{align}\label{P_le_S1}
P&\le\bigg(A_2he^{-A_1h}\frac{e^{-(A_1+A_4)T}-1}{e^{-(A_1+A_4)h}-1}+\frac{A_3}{A_6}e^{-A_1h}\bigg)S,
\end{align}
and
\begin{align}\label{S_le_P1}
	S&\le e^{A_4T}(1+\lambda_3^{-1})\bE|g(V_T,X^{\pi}_T)-Y^{\pi}_T|^2
	+\bigg(g_x(1+\lambda_3)e^{(A_1+A_4)T}+A_5h\frac{e^{(A_1+A_4)T}-1}
	{e^{(A_1+A_4)h}-1}\bigg)P.
\end{align}
Set $$A(h)=\left( A_2he^{-A_1h}\frac{e^{-(A_1+A_4)T}-1}{e^{-(A_1+A_4)h}-1}+\frac{A_3}{A_6}e^{-A_1h}\right)
\left(g_x(1+\lambda_3)e^{(A_1+A_4)T}+A_5h\frac{e^{(A_1+A_4)T}-1}{e^{(A_1+A_4)h}-1}\right). $$
If $A(h)<1$, then we can have by combining estimates from \eqref{P_le_S1} and \eqref{S_le_P1} that
\begin{align}
P&\le\big(1-A(h)\big)^{-1}e^{A_4T}(1+\lambda_3^{-1})\bE|g(V_T,X^{\pi}_T)-Y^{\pi}_T|^2
\left( A_2he^{-A_1h}\frac{e^{-(A_1+A_4)T}-1}{e^{-(A_1+A_4)h}-1}+\frac{A_3}{A_6}e^{-A_1h}\right),
\end{align}
and
\begin{align}
	S&\le\big(1-A(h)\big)^{-1}e^{A_4T}(1+\lambda_3^{-1})\bE|g(V_T,X^{\pi}_T)-Y^{\pi}_T|^2.
\end{align}
Now, define
\begin{align}
\overline{P}&:=\max_{0\le n\le N}e^{-\overline{A}_1nh}\bE|\delta X_n|^2, \nonumber\\
\overline{S}&:=\max\left\{ \max_{0\le n\le N}\big(e^{\overline{A}_4nh}\bE|\delta Y_n|^2\big),\,\,
			\sum_{i=0}^{N-1}e^{\overline{A}_4ih}\overline{A}_6h\bE|\delta Z_i|^2\right\}.\nonumber
\end{align}
Recall that $A_i\rightarrow \overline{A}_i$ as $h\to 0$ for $i=1,\cdots,6,$ and that
$$\lim_{h\to 0}A(h)=\left( \overline{A}_2\frac{1-e^{-(\overline{A}_1+\overline{A}_4)T}}{\overline{A}_1+\overline{A}_4}+\frac{\overline{A}_3}{\overline{A}_6}\right)
\left(g_x(1+\lambda_3)e^{(\overline{A}_1+\overline{A}_4)T}+\overline{A}_5\frac{e^{(\overline{A}_1+\overline{A}_4)T}-1}{\overline{A}_1+\overline{A}_4}\right).
$$
When $\overline{A}_0<1$, comparing $\lim_{h\to 0} A(h)$ and $\overline{A_0}$, we know that, for any $\epsilon>0$, there exists $\lambda_3>0$ and sufficiently small $h$ such that
\begin{align}
\overline{P}\le&(1+\epsilon)\big(1-\overline{A}_0\big)^{-1}e^{\overline{A}_4T}(1+\lambda_3^{-1})\bigg(\overline{A}_2\frac{1-e^{-(\overline{A}_1+\overline{A}_4)T}}{\overline{A}_1+\overline{A}_4}+\frac{\overline{A}_3}{\overline{A}_6}\bigg)\bE|g(V_T,X^{\pi}_T)-Y^{\pi}_T|^2,
\end{align}
and 
\begin{align}
	\overline{S}\le&(1+\epsilon)\big(1-\overline{A}_0\big)^{-1}e^{\overline{A}_4T}(1+\lambda_3^{-1})\bE|g(V_T,X^{\pi}_T)-Y^{\pi}_T|^2.
\end{align}
By setting $\epsilon=1$ and choosing suitable $\lambda_3$, we obtain estimates of $\bE|\delta X_n|^2, \bE|\delta Y_n|^2$ and $\bE|\delta Z_n|^2$ as
\begin{equation*}
\max_{0\le n\le N}\bE|\delta X_n|^2\le e^{\overline{A}_1T\vee 0}\overline{P}\le C(\lambda_1,\lambda_2)\bE|g(V_T,X^{\pi}_T)-Y^{\pi}_T|^2,
\end{equation*}
\begin{equation*}
\max\left\{ \max_{0\le n\le N}\bE|\delta Y_n|^2,\,\,
\sum_{i=0}^{N-1}\bE|\delta Z_i|^2h\right\}
\le e^{-\overline{A}_4T\vee 0}(1+\overline{A}_6^{-1})\overline{S}
\le C(\lambda_1,\lambda_2)\bE|g(V_T,X^{\pi}_T)-Y^{\pi}_T|^2.
\end{equation*}
Finally, combining these bounds with the result of Theorem \ref{Rmain_convergence}, we complete the proof.
\end{proof}

\section{A Convergence Analysis}

In this section, we establish an estimate for the minimized objective function associated with the stochastic optimization problem (\ref{scheme_objective}-\ref{scheme_system}). This result serves as a theoretical justification for employing our numerical scheme in the context of non-Markovian forward-backward stochastic differential equations (FBSDEs).

\begin{thm}\label{convergence_objective_function}
    Suppose all assumptions of Theorem \ref{stability_FBSDE} hold. Let $X^{\pi}_T, Y^{\pi}_0,Y^{\pi}_T, \{Z^{\pi}_{t_i}\}_{0\le i\le N-1}$ be defined as in \eqref{scheme_system}. Then, there exists a constant $C>0$, independent of $h,d$ and $m$, such that for sufficiently small $h$,

\begin{align*}
	\inf _{\cY_0,\cZ_i \in \mathcal{NN}(\cdot)}\bE&|g(V_T,X^{\pi}_T)-Y^{\pi}_T|^2\le C\Bigg[\big(I_{0}^2+\bE |x_0|^2\big)h+\rho(h)+ h \bE \int_{0}^{T} |Z_s|^2ds+\sum_{i=0}^{N-1}E^i_z\\& 
    +\inf _{\cY_0,\cZ_i \in \mathcal{NN}(\cdot)}\left(\sum_{i=0}^{N-1}\bE|\tilde{Z}_{t_i}-\cZ_i\big(X^\pi_{t_i},\{W_s,V_s:0\le s\le t_i\}\big)|^2h+\bE|Y_0-\cY_0(x_0)|^2\right)\Bigg],
\end{align*}

where
\begin{align*}
    \tilde{Z}_{t_i}&=h^{-1}\bE\Big[\int_{t_i}^{t_{i+1}}Z_t dt|\mathcal{F}_{t_i}\Big],\\
    E^i_z&=\int_{t_i}^{t_{i+1}}\bE|Z_t-\tilde{Z}_{t_i}|^2dt.
\end{align*}
\end{thm}

\begin{proof}
Consider the piecewise constant Euler-Maruyama approximation defined on each interval $[t_i,t_{i+1})$. Define the numerical processes $X^{\pi}_t$ and $Y^{\pi}_t$ by:
\begin{equation}\label{th41:XY_disc}
\left\{\begin{array}{l}
{X^{\pi}_t=X^{\pi}_{t_i}+b(t_i,V_{t_i},X^{\pi}_{t_i},Y^{\pi}_{t_i},Z^{\pi}_{t_i})(t-t_i)+\sigma(t_i,V_{t_i},X^{\pi}_{t_i},Y^{\pi}_{t_i},Z^{\pi}_{t_i})(W_t-W_{t_i})},\\
{Y^{\pi}_t=Y^{\pi}_{t_i}-f(t_i,V_{t_i},X^{\pi}_{t_i},Y^{\pi}_{t_i},Z^{\pi}_{t_i})(t-t_i)+Z^{\pi}_{t_i}(W_t-W_{t_i})}.
\end{array}\right.
\end{equation}
Now, define the error processes by
$\delta X_t=X_t-X^{\pi}_t, \delta Y_t=Y_t-Y^{\pi}_t, \delta Z_t=Z_t-Z^{\pi}_{t_i} \text{ for }t\in[t_i,t_{i+1})$ and the coefficient differences by
\begin{equation}
\left\{\begin{array}{l}
{\delta b_t=b(t,V_t,X_t,Y_t,Z_t)-b(t_i,V_{t_i},X^{\pi}_{t_i},Y^{\pi}_{t_i},Z^{\pi}_{t_i})},\\
{\delta \sigma_t=\sigma(t,V_t,X_t,Y_t,Z_t)-\sigma(t_i,V_{t_i},X^{\pi}_{t_i},Y^{\pi}_{t_i},Z^{\pi}_{t_i})},\\
{\delta f_t=f(t,V_t,X_t,Y_t,Z_t)-f(t_i,V_{t_i},X^{\pi}_{t_i},Y^{\pi}_{t_i},Z^{\pi}_{t_i})}.
\end{array}\right.
\end{equation}
Then we have from \eqref{th41:XY_disc},
\begin{equation}
\left\{\begin{array}{l}
{d(\delta X_t)=\delta b_t dt+\delta\sigma_t dW_t},\\
{d(\delta Y_t)=-\delta f_t dt+\delta Z_t dW_t}.
\end{array}\right.
\end{equation}
Next, use of It$\hat{o}$'s formula yields that
\begin{equation}
\left\{\begin{array}{l}
{d|\delta X_t|^2=\big(2|\delta b_t|\cdot|\delta X_t|+|\delta \sigma_t|^2\big)dt+2|\delta X_t|\cdot|\delta\sigma_t| dW_t},\\
{d|\delta Y_t|^2=\big(-2|\delta f_t|\cdot|\delta Y_t|+|\delta Z_t|^2\big) dt+2|\delta Y_t|\cdot|\delta Z_t| dW_t}.
\end{array}\right.
\end{equation}
Now by taking expectations on both sides we get
\begin{equation}\label{diff_eqn}
\left\{\begin{array}{l}
{\bE|\delta X_t|^2=\bE|\delta X_{t_i}|^2+\bE\int_{t_i}^{t} \big(2|\delta b_s|\cdot|\delta X_s|+|\delta \sigma_s|^2\big)ds},\\
{\bE|\delta Y_t|^2=\bE|\delta Y_{t_i}|^2+\bE\int_{t_i}^{t}\big(-2|\delta f_s|\cdot|\delta Y_s|+|\delta Z_s|^2\big)ds}.
\end{array}\right.
\end{equation}
Then using Assumption \ref{ass-1}, for any $\lambda_5,\lambda_6>0$, we have from the 1st equation of (\ref{diff_eqn}) that
\begin{align}\label{delta_X}
\bE|\delta X_t|^2\le \;& \bE|\delta X_{t_i}|^2+\int_{t_i}^{t}\big(\lambda_5\bE|\delta X_s|^2+\lambda_5^{-1}\bE|\delta b_s|^2+\bE|\delta \sigma_s|^2\big)ds\nonumber\\
\le\; &\bE|\delta X_{t_i}|^2+\lambda_5\int_{t_i}^{t}\bE|\delta X_s|^2ds+\int_{t_i}^{t}(\lambda_5^{-1}+1)\rho(|s-t_i|)ds\nonumber\\
&+\int_{t_i}^{t}\Big((K\lambda_5^{-1}+\sigma_x)\bE|X_s-X^{\pi}_{t_i}|^2+(b_y\lambda_5^{-1}+\sigma_y)\bE|Y_s-Y^{\pi}_{t_i}|^2\nonumber\\
&+(b_z\lambda_5^{-1}+\sigma_z)\bE|Z_s-Z^{\pi}_{t_i}|^2\Big)ds.
\end{align}
Note that, for any $\epsilon_1,\epsilon_2>0$,
\begin{equation}\label{est:t41XY}
\left \{\begin{array}{l}
{\bE| X_s-X^{\pi}_{t_i}|^2\le(1+\epsilon_1)\bE|\delta X_{t_i}|^2+(1+\epsilon_1^{-1})\bE|X_s-X_{t_i}|^2},\\
{\bE| Y_s-Y^{\pi}_{t_i}|^2\le(1+\epsilon_2)\bE|\delta Y_{t_i}|^2+(1+\epsilon_2^{-1})\bE|Y_s-Y_{t_i}|^2},
\end{array}\right.
\end{equation}
and we choose $\epsilon_1=\lambda_6(K\lambda_5^{-1}+\sigma_x)^{-1}, \epsilon_2=\lambda_6(b_y\lambda_5^{-1}+\sigma_y)^{-1}$.\\
Also note, Theorem \ref{stability_FBSDE} yields
\begin{align}\label{est:t41XY2}
	\sup_{s\in[t_i,t_{i+1})}&\Big(\bE|X_s-X_{t_i}|^2+\bE|Y_s-Y_{t_i}|^2\Big)\nonumber\\
    &\le C\big(I_{0}^2+\bE |x_0|^2\big)h + C \bE\left[ \int_{t_i}^{t_{i+1}} |Z_s|^2ds+ \int_{t_i}^{t_{i+1}} |\sigma|^2(s,V_s,0,0,0)ds \right].
\end{align}
Now, first using the estimates from \eqref{est:t41XY} into \eqref{delta_X}, and then using the estimate from \eqref{est:t41XY2}, we get
\begin{align}\label{est:t41X5}
\bE|\delta X_t|^2\le\;&\bE|\delta X_{t_i}|^2+\lambda_5\int_{t_i}^{t}\bE|\delta X_s|^2ds+\int_{t_i}^{t}(\lambda_5^{-1}+1)\rho(|s-t_i|)ds\nonumber\\
&+\int_{t_i}^{t}\Big((K\lambda_5^{-1}+\sigma_x)(1+\epsilon_1)\bE|\delta X_{t_i}|^2+(K\lambda_5^{-1}+\sigma_x)(1+\epsilon_1^{-1})\bE|X_s-X_{t_i}|^2\nonumber\\
&+(b_y\lambda_5^{-1}+\sigma_y)(1+\epsilon_2)\bE|\delta Y_{t_i}|^2+(b_y\lambda_5^{-1}+\sigma_y)(1+\epsilon_2^{-1})\bE|Y_s-Y_{t_i}|^2\nonumber\\
&+(b_z\lambda_5^{-1}+\sigma_z)\bE|\delta Z_s|^2\Big)ds\nonumber\\
\le\;&\big(1+(\lambda_6+K\lambda_5^{-1}+\sigma_x)h\big)\bE|\delta X_{t_i}|^2+(\lambda_6+b_y\lambda_5^{-1}+\sigma_y)h\bE|\delta Y_{t_i}|^2\nonumber\\
&+\int_{t_i}^{t}(b_z\lambda_5^{-1}+\sigma_z)\bE|\delta Z_s|^2ds +\lambda_5\int_{t_i}^{t}\bE|\delta X_s|^2ds+C\Big(\big(I_{0}^2+\bE |x_0|^2\big)h+\rho(h)\Big)h\nonumber\\
& + Ch \bE \left[\int_{t_i}^{t_{i+1}} |Z_s|^2ds+\int_{t_i}^{t_{i+1}}|\sigma|^2(s,V_s,0,0,0)ds\right].
\end{align}
Then applying  Gr$\ddot{o}$nwall inequality into \eqref{est:t41X5} and using the estimate 
$$\int_{t_i}^{t_{i+1}}\bE|\delta Z_t|^2dt\le(1+\epsilon_3)\bE|\delta\tilde{Z}_{t_i}|^2h+(1+\epsilon_3^{-1})E^i_z,$$
with $\epsilon_3=(b_z\lambda_5^{-1}+\sigma_z)^{-1}\lambda_6$, we obtain
\begin{align}\label{est:t41X6}
\bE&|\delta X_{t_{i+1}}|^2\nonumber\\
&\le 
	e^{\lambda_5h}\bigg(\big(1+(\lambda_6+K\lambda_5^{-1}+\sigma_x)h\big)\bE|\delta X_{t_i}|^2+(\lambda_6+b_y\lambda_5^{-1}+\sigma_y)h\bE|\delta Y_{t_i}|^2\nonumber\\
&+\int_{t_i}^{t_{i+1}}(b_z\lambda_5^{-1}+\sigma_z)\bE|\delta Z_s|^2ds
	+C\Big(\big(I_{0}^2+\bE |x_0|^2\big)h+\rho(h)\Big)h\nonumber\\
    &+ Ch \bE \left[\int_{t_i}^{t_{i+1}} |Z_s|^2ds+\int_{t_i}^{t_{i+1}}|\sigma|^2(s,V_s,0,0,0)ds\right]\bigg)
	\nonumber\\
	&\le e^{(\lambda_5+\lambda_6+K\lambda_5^{-1}+\sigma_x)h}\bE|\delta X_{t_i}|^2+e^{\lambda_5h}(b_y\lambda_5^{-1}+\lambda_6+\sigma_y)h\bE|\delta Y_{t_i}|^2+e^{\lambda_5h}(b_z\lambda_5^{-1}+\lambda_6+\sigma_z)\bE|\delta\tilde{Z}_{t_i}|^2h
	\nonumber\\
&
+CE^i_z+C\Big(\big(I_{0}^2+\bE |x_0|^2\big)h+\rho(h)\Big)h+ Ch \bE \left[\int_{t_i}^{t_{i+1}} |Z_s|^2ds+\int_{t_i}^{t_{i+1}}|\sigma|^2(s,V_s,0,0,0)ds\right]
\nonumber\\
&\le e^{A_7h}\bE|\delta X_{t_i}|^2+A_8\bE|\delta Y_{t_i}|^2h+A_{9}\bE|\delta \tilde{Z}_{t_i}|^2h+CE^i_z+C\Big(\big(I_{0}^2+\bE |x_0|^2\big)h+\rho(h)\Big)h\nonumber\\
&+ Ch \bE \left[\int_{t_i}^{t_{i+1}} |Z_s|^2ds+\int_{t_i}^{t_{i+1}}|\sigma|^2(s,V_s,0,0,0)ds\right] ,
\end{align}
where $A_7=K\lambda_5^{-1}+\sigma_x+\lambda_5+\lambda_6, A_8=b_y\lambda_5^{-1}+2\lambda_6+\sigma_y, A_{9}=b_z\lambda_5^{-1}+2\lambda_6+\sigma_z,$ and $h$ is sufficiently small.\\

Similarly, for any $\lambda_5,\lambda_6>0$, we have from the 2nd equation of (\ref{diff_eqn}) that
\begin{align}\label{est:t41Y5}
\bE|\delta Y_t|^2\le\;&\bE|\delta Y_{t_i}|^2+\int_{t_i}^{t}\Big(\lambda_5\bE|\delta Y_s|^2+\lambda_5^{-1}\bE|\delta f_s|^2+\bE|\delta Z_s|^2\Big)ds\nonumber\\
\le\;&\bE|\delta Y_{t_i}|^2+\lambda_5\int_{t_i}^{t}\bE|\delta Y_s|^2ds+\int_{t_i}^{t}\lambda_5^{-1}\Big(f_x\bE|X_s-X^{\pi}_{t_i}|^2+K\bE|Y_s-Y^{\pi}_{t_i}|^2\Big)ds\nonumber\\
&+\int_{t_i}^{t}\lambda_5^{-1}\rho(|s-t_i|)ds+(1+f_z\lambda_5^{-1})\int_{t_i}^{t} \bE|\delta Z_s|^2ds\nonumber\\
\le&\big(1+(K\lambda_5^{-1}+\lambda_6)h\big)\bE|\delta Y_{t_i}|^2+\lambda_5\int_{t_i}^{t}\bE|\delta Y_s|^2ds+(f_x\lambda_5^{-1}+\lambda_6)\bE|\delta X_{t_i}|^2h\nonumber\\
&+(1+f_z\lambda_5^{-1})\int_{t_i}^{t} \bE|\delta Z_s|^2ds+\lambda_5^{-1}\rho(h)h+C\int_{t_i}^t\bE\bigg[|Y_s-Y_{t_i}|^2+|X_s-X_{t_i}|^2\bigg]ds\nonumber\\
\le\; & \big(1+(K\lambda_5^{-1}+\lambda_6)h\big)\bE|\delta Y_{t_i}|^2+\lambda_5\int_{t_i}^{t}\bE|\delta Y_s|^2ds+(f_x\lambda_5^{-1}+\lambda_6)\bE|\delta X_{t_i}|^2h\nonumber\\
&+(1+f_z\lambda_5^{-1})\int_{t_i}^{t} \bE|\delta Z_s|^2ds+C\Big(\big(I_{0}^2+\bE |x_0|^2\big)h+\rho(h)\Big)h\nonumber\\
&+ Ch \bE \left[\int_{t_i}^{t_{i+1}} |Z_s|^2ds+\int_{t_i}^{t_{i+1}}|\sigma|^2(s,V_s,0,0,0)ds\right].
\end{align}
Then once again using Gr$\ddot{o}$nwall inequality, for sufficiently small $h$, we have from \eqref{est:t41Y5} that
\begin{align}\label{est:t41Y6}
\bE|\delta Y_{t_{i+1}}|^2\le &e^{A_{10}h}\bE|\delta Y_{t_i}|^2+A_{11}\bE|\delta X_{t_i}|^2h+A_{12}\bE|\delta\tilde{Z}_{t_i}|^2h+C\Big(\big(I_{0}^2+\bE |x_0|^2\big)h+\rho(h)\Big)h\nonumber\\
&+ Ch \bE \left[\int_{t_i}^{t_{i+1}} |Z_s|^2ds+\int_{t_i}^{t_{i+1}}|\sigma|^2(s,V_s,0,0,0)ds\right],
\end{align}
where $A_{10}=\lambda_5+K\lambda_5^{-1}+\lambda_6, A_{11}=f_x\lambda_5^{-1}+2\lambda_6, A_{12}=1+2\lambda_6+f_z\lambda_5^{-1}$. \\
Now define $M_i=max\{\bE|\delta X_i|^2, \bE|\delta Y_i|^2\}, 0\le i\le N,$ and then from estimates \eqref{est:t41X6} and \eqref{est:t41Y6} we have
\begin{align}
M_{i+1}
&\le \bigg(e^{max\{A_7,A_{10}\}h}+max\{A_8,A_{11}\}h\bigg)M_i+max\{A_{9},A_{12}\}\bE|\delta\tilde{Z}_{t_i}|^2h +CE^i_z\nonumber\\
	&\quad
	+C\Big(\big(I_{0}^2+\bE |x_0|^2\big)h+\rho(h)\Big)h+ Ch \bE \left[\int_{t_i}^{t_{i+1}} |Z_s|^2ds+\int_{t_i}^{t_{i+1}}|\sigma|^2(s,V_s,0,0,0)ds\right]
\nonumber\\
&\le 
	e^{(max\{A_7,A_{10}\}+max\{A_8,A_{11}\})h}M_i+max\{A_{9},A_{12}\}\bE|\delta\tilde{Z}_{t_i}|^2h+CE^i_z
\nonumber\\
&\quad
+C\Big(\big(I_{0}^2+\bE |x_0|^2\big)h+\rho(h)\Big)h+ Ch \bE \left[\int_{t_i}^{t_{i+1}} |Z_s|^2ds+\int_{t_i}^{t_{i+1}}|\sigma|^2(s,V_s,0,0,0)ds\right].
\end{align}
By letting $A_{13}=max\{A_7,A_{10}\}+max\{A_8,A_{11}\}, A_{14}=max\{A_{9},A_{12}\}$ we further have
\begin{align}
M_{i+1}\le & e^{A_{13}h}M_i+A_{14}\bE|\delta\tilde{Z}_{t_i}|^2h+CE^i_z+C\Big(\big(I_{0}^2+\bE |x_0|^2\big)h+\rho(h)\Big)h\nonumber\\
&+ Ch \bE \left[\int_{t_i}^{t_{i+1}} |Z_s|^2ds+\int_{t_i}^{t_{i+1}}|\sigma|^2(s,V_s,0,0,0)ds\right].
\end{align}
Since $M_0=\bE|Y_0-Y^{\pi}_0|^2$, we can also have
\begin{align}
M_N\le & A_{14}e^{A_{13}T}\sum_{i=0}^{N-1}\bE|\delta\tilde{Z}_{t_i}|^2h\nonumber\\
&+C(T)\bigg(\big(I_{0}^2+\bE |x_0|^2\big)h+\rho(h)+ \bE|Y_0-Y^{\pi}_0|^2
+h \bE \int_{0}^{T} |Z_s|^2ds+\sum_{i=0}^{N-1}E^i_z\bigg).
\end{align}
Note that, for $\bar{K}=\max\{K,f_x,f_z,b_y,b_z,\sigma_x,\sigma_y,\sigma_z\}$:
$$A_{14}\le 1+\bar{K}\lambda_5^{-1}+\bar{K}+2\lambda_6,$$
$$A_{13}\le 2\bar{K}+2\bar{K}\lambda_5^{-1}+\lambda_5+3\lambda_6.$$
Given any $\lambda_4>0$, we can choose $\lambda_6$ small enough such that
$$A_{14}e^{A_{13}T}\le(1+\lambda_4)(\bar{K}\lambda_5^{-1}+\bar{K}+1)e^{(2\bar{K}\lambda_5^{-1}+2\bar{K}+\lambda_5)T}.$$
Thus
\begin{align}\label{est:M_N}
M_N\le&(1+\lambda_4)(\bar{K}\lambda_5^{-1}+\bar{K}+1)e^{(2\bar{K}\lambda_5^{-1}+2\bar{K}+\lambda_5)T}\sum_{i=0}^{N-1}\bE|\delta\tilde{Z}_{t_i}|^2h\nonumber\\
&+C(T)\bigg(\big(I_{0}^2+\bE |x_0|^2\big)h+\rho(h)+\bE|Y_0-Y^{\pi}_0|^2+ h \bE \int_{0}^{T} |Z_s|^2ds+\sum_{i=0}^{N-1}E^i_z\bigg),
\end{align}
Now, from the decomposition of the objective function, we have
\begin{align}\label{est:obj_T}
\bE|g(V_T, X^{\pi}_T)-Y^{\pi}_T|^2=&\bE|g(V_T, X^{\pi}_T)-g(V_T,X_T)+Y_T-Y^{\pi}_T|^2\nonumber\\
\le&\big(1+(\sqrt{g_x})^{-1}\big)\bE|g(V_T,X^{\pi}_T)-g(V_T,X_T)|^2+(1+\sqrt{g_x})\bE|\delta Y_N|^2\nonumber\\
\le&(g_x+\sqrt{g_x})\bE|\delta X_N|^2+(1+\sqrt{g_x})\bE|\delta Y_N|^2\nonumber\\
\le&(1+\sqrt{g_x})^2M_N.
\end{align}
Then, by substituting the estimate for $M_N$ from \eqref{est:M_N} into \eqref{est:obj_T} and choosing $\lambda_5=argmin_{\lambda\in\mathbb{R}^{+}}H(\lambda)$, where
$$H(\lambda)=(1+\sqrt{g_x})^2e^{(2\bar{K}\lambda^{-1}+2\bar{K}+\lambda)T}(\bar{K}\lambda^{-1}+\bar{K}+1),$$
we can have the following estimate
\begin{align}\label{est_final_sec4}
	\bE&|g(V_T,X^{\pi}_T)-Y^{\pi}_T|^2\le (1+\lambda_4)H_{min}\sum_{i=0}^{N-1}\bE|\delta \tilde{Z}_{t_i}|^2h\nonumber\\
	&+C(T)\bigg(\big(I_{0}^2+\bE |x_0|^2\big)h+\rho(h)+\bE|Y_0-Y^{\pi}_0|^2+ h \bE \int_{0}^{T} |Z_s|^2ds+\sum_{i=0}^{N-1}E^i_z\bigg),
\end{align}
with $$H_{min}=\min_{\lambda\in\mathbb{R}^{+}}H(\lambda).$$
Finally the proof is complete by using the fact $Y^{\pi}_0=\cY_0(x_0)$ and the representation $Z_{t_i}^{\pi}=\cZ_i\big(X^{\pi}_{t_i},\{W_s,V_s;0\le s\le t_i\}\big)$, and then taking the infimum on both sides of the equation \eqref{est_final_sec4} with respect to $\cY_0$ and $\cZ_i,i=0,\cdots,N-1$, within neural network function classes of appropriate structure.

\end{proof}

%%%%%%%%%%%%%%%%%%%%%%%%%%%%%%%%%%%%%%%%%%%%%
%%%%%%%%%%%%%%%%%%%%%%%%%%%%%%%%%%%%%%%%%%%%%%%%%%%%%%%%%%%%%%%%%%%%%%%%%%%%%%%%%%%%%%%%%%
%%%%%%%%%%%%%%%%%%%%%%%%%%%%%%%%%%%%%%%%%%%%%%%%%%%%%%%%

\section{Algorithm for Coupled non-Markovian FBSDEs} \label{algorithm_Section}

We consider the coupled forward-backward stochastic differential equation (FBSDE) given in \eqref{maineq}, explicitly written as
\begin{equation}
	\begin{cases}
		dX_t &= b(t,V_t,X_t,Y_t,Z_t)dt + \sigma(t,V_t,X_t,Y_t,Z_t)dW_t, \ t\in[0,T],\\
		X_0&=x_0,\\
		-dY_t &=f(t,V_t,X_t,Y_t,Z_t)dt - Z_t dW_t, \ t\in[0,T], \\
		Y_T &= g(V_T,X_T).
	\end{cases}
\end{equation}
Under the Markovian framework, deep learning-based algorithms for high-dimensional coupled FBSDEs (or FBSPDEs) have already demonstrated remarkable efficiency, as we discussed in Section 1. Here, we extend the setting to the non-Markovian framework. 

Inspired by the fact that, in a non-Markovian setting, the backward processes $(Y_t,Z_t)$ can be represented as random functionals of the forward process $X_t$, we aim to approximate $(Y_t,Z_t)$ as functions of both $X_t$ and and the path information $\{W_s,V_s:0\le s\le t\}$. More precisely, we assume
$$Y_t=\mathcal{Y}(t,X_t,\{W_s,V_s:0\le s\le t\}),\;Z_t=\mathcal{Z}(t,X_t,\{W_s,V_s:0\le s\le t\}),$$ 
for some deterministic functionals $\mathcal{Y}$ and $\mathcal{Z}$. This representation forms a key component of our deep learning-based approximation scheme for the backward processes.

In practice, deep neural networks are employed to approximate the unknown functionals $\mathcal{Y},\mathcal{Z}$ by reformulating the original FBSDE problem as a stochastic optimization problem. Here, we present two deep learning-based algorithms for solving the coupled non-Markovian FBSDE, both inspired by the framework of \cite{SJiSPengYPengXZhang2020ThreeAlgorithmsforSolvHighDimCoupledFBSDEsDeepLearn}.

To describe these algorithms, we consider a uniform partition $\pi:{t_0=0<t_1<\dots<t_N=T}$ of the time interval $[0,T]$. We denote by $X^{\pi}$ and $Y^{\pi}$ the discrete-time approximations of the processes $X$ and $Y$, respectively, obtained via the Euler scheme on the time grid $\pi$. We also define the Brownian motion increments as $\Delta W_{t_j}:=W_{t_{j+1}}-W_{t_j}$ for $j=0,1,\cdots,N-1$. 

The forward representation of the backward equation in \eqref{maineq} is given by:
\begin{equation}\label{fwrepbw}
Y_t=Y_{0}-\int_{0}^{t}f(s,V_s,X_s,Y_s,Z_s)\mathrm{d}s+\int_{0}^{t}Z_s\mathrm{d}W_s,\quad 0\le t\le T.
\end{equation}
This representation forms the foundation for the following two numerical algorithms designed to approximate the solution of coupled non-Markovian FBSDEs. Now, we describe the two algorithms.

\subsection{Algorithm 1}

The steps are as follows:
\begin{itemize}
\item Initialization starts with estimations $\cY_0$ and $\cZ_0$ of $Y_{0}$ and 
$Z_{0}$ respectively and then calculate $X_{t_1}^{\pi}$ and $Y_{t_1}^{\pi}$ by using Euler scheme as
\begin{equation*}
	\begin{split}
		X_{t_1}^{\pi}&=X_{0}+b(t_0,V_{t_0},X_{0},\cY_0,\cZ_0 )\Delta t_0+\sigma(t_0,V_{t_0},X_{0},\cY_0,\cZ_0)\Delta W_{t_0},\\
		Y_{t_1}^{\pi}&=\cY_0-f(t_0,V_{t_0},X_{0},\cY_0,\cZ_0 )\Delta t_0+\cZ_0\Delta W_{t_0}.
	\end{split}
\end{equation*}

\item For each $j=1,2,\cdots,N-1$, a neural network $\cZ_j^{\mathcal{N}}(\cdot;\theta_{j})$ approximate $Z_{t_j}=\cZ\big(t_j,X_{t_j}^{\pi},\{W_{t_i},V_{t_i}\}_{i=0}^j\big)$ and Euler scheme calculates $X_{t_{j+1}}^{\pi}$ and $Y_{t_{j+1}}^{\pi}$ as
\begin{equation*}
	\begin{split}
		X_{t_{j+1}}^{\pi}&=X_{t_j}^{\pi}+b\bigg(t_j,V_{t_j},X_{t_j}^{\pi},Y_{t_j}^{\pi},\cZ_j^{\mathcal{N}}\big(X_{t_j}^{\pi},\{W_{t_i},V_{t_i}\}_{i=0}^j;\theta_{j}\big) \bigg)\Delta t_j\\
		&+\sigma\bigg(t_j,V_{t_j},X_{t_j}^{\pi},Y_{t_j}^{\pi},\cZ_j^{\mathcal{N}}\big(X_{t_j}^{\pi},\{W_{t_i},V_{t_i}\}_{i=0}^j;\theta_{j}\big)\bigg)\Delta W_{t_j},\\
		Y_{t_{j+1}}^{\pi}&=Y_{t_{j}}^{\pi}-f\bigg(t_j,V_{t_j},X_{t_j}^{\pi},Y_{t_j}^{\pi},\cZ_j^{\mathcal{N}}\big(	X_{t_j}^{\pi},\{W_{t_i},V_{t_i}\}_{i=0}^j;\theta_{j}\big) \bigg)\Delta t_j+\cZ_j^{\mathcal{N}}\big(	X_{t_j}^{\pi},\{W_{t_i},V_{t_i}\}_{i=0}^j;\theta_{j}\big)\Delta W_{t_j}.
	\end{split}
\end{equation*}

\item Finally, the scheme is to optimize the loss function
\begin{equation*}
	\widehat{\mathfrak{L}}\big(\cY_0,\cZ_0,\theta_{1},\theta_{2},\cdots,\theta_{N-1}\big):=\bE\bigg|Y_{t_{N}}^{\pi}-g(V_{t_N},X_{t_{N}}^{\pi})\bigg|^2,
\end{equation*}
over all $\theta=(\cY_0,\cZ_0,\theta_{1},\theta_{2},\cdots,\theta_{N-1})$, and if 
%$\theta^*=(\cY^*_0,\cZ^*_0,\theta^*_{1,1},\theta^*_{2,1},\cdots,\theta^*_{1,J},\theta^*_{2,J})$ if
\begin{equation*}
	\theta^*=(\cY^*_0,\cZ^*_0,\theta^*_{1},\theta^*_{2},\cdots,\theta^*_{N-1})\in \arg \min_{\theta\in\Theta} \widehat{\mathfrak{L}}(\theta),
\end{equation*}
then $\cY_0^*$ is the desired approximation of $Y_{0}$ by this method.
\end{itemize}

%We refer to \textbf{cite JiPeng2020CoupledBSDE} for various numerical examples, while the reader may refer to \textbf{cite Han&Long2019coupleddeep} for an alternative algorithm for a class of coupled Markovian FBSDEs with both convergence analysis and numerical examples.

\subsection{Algorithm 2}

The steps are as follows:

\begin{itemize}
\item Initialize with estimations $\cY_0$ and $\cZ_0$ of $Y_{0}$ and 
$Z_{0}$ respectively and calculate $X_{t_1}^{\pi}$ and $Y_{t_1}^{\pi}$ by using Euler scheme as
\begin{equation*}
	\begin{split}
		X_{t_1}^{\pi}&=X_{0}+b(t_0,V_{t_0},X_{0},\cY_0,\cZ_0 )\Delta t_0+\sigma(t_0,V_{t_0},X_{0},\cY_0,\cZ_0)\Delta W_{t_0},\\
		Y_{t_1}^{\pi}&=\cY_0-f(t_0,V_{t_0},X_{0},\cY_0,\cZ_0 )\Delta t_0+\cZ_0\Delta W_{t_0}.
	\end{split}
\end{equation*}
\item For each $j=1,2,\cdots,N-1$, a neural network $\cX_j^{\cN}(\cdot;\theta_j)=\big(\cY_j^{\mathcal{N}}(\cdot;\theta_{j}), \cZ_j^{\mathcal{N}}(\cdot;\theta_{j})\big)$ approximates $Y_{t_j}=\cY\big(t_j,X_{t_j}^{\pi},\{W_{t_i},V_{t_i}\}_{i=0}^j\big)$ and $Z_{t_j}=\cZ\big(t_j,X_{t_j}^{\pi},\{W_{t_i},V_{t_i}\}_{i=0}^j\big)$ 
with the associated local loss function, which is defined as
\begin{equation*}
	\mathfrak{L}_j:=\Delta t_{j-1}\cdot\bE\bigg|Y_{t_j}^{\pi}-\cY_j^{\mathcal{N}}\big(X_{t_j}^{\pi},\{W_{t_i},V_{t_i}\}_{i=0}^j;\theta_{j}\big)\bigg|^2.
\end{equation*}

Then Euler scheme calculates $X_{t_{j+1}}^{\pi}$ and $Y_{t_{j+1}}^{\pi}$ as
\begin{equation*}
	\begin{split}
		X_{t_{j+1}}^{\pi}&=X_{t_j}^{\pi}+b\bigg(t_j,V_{t_j},X_{t_j}^{\pi},\cY_j^{\mathcal{N}}\big(	X_{t_j}^{\pi},\{W_{t_i},V_{t_i}\}_{i=0}^j;\theta_{j}\big),\cZ_j^{\mathcal{N}}\big(X_{t_j}^{\pi},\{W_{t_i},V_{t_i}\}_{i=0}^j;\theta_{j}\big) \bigg)\Delta t_j\\
		&+\sigma\bigg(t_j,V_{t_j},X_{t_j}^{\pi},\cY_j^{\mathcal{N}}\big(	X_{t_j}^{\pi},\{W_{t_i},V_{t_i}\}_{i=0}^j;\theta_{j}\big),\cZ_j^{\mathcal{N}}\big(X_{t_j}^{\pi},\{W_{t_i},V_{t_i}\}_{i=0}^j;\theta_{j}\big)\bigg)\Delta W_{t_j},\\
		Y_{t_{j+1}}^{\pi}&=Y_{t_{j}}^{\pi}-f\bigg(t_j,V_{t_j},X_{t_j}^{\pi},\cY_j^{\mathcal{N}}\big(	X_{t_j}^{\pi},\{W_{t_i},V_{t_i}\}_{i=0}^j;\theta_{j}\big),\cZ_j^{\mathcal{N}}\big(	X_{t_j}^{\pi},\{W_{t_i},V_{t_i}\}_{i=0}^j;\theta_{j}\big) \bigg)\Delta t_j\\
		&+\cZ_j^{\mathcal{N}}\big(	X_{t_j}^{\pi},\{W_{t_i},V_{t_i}\}_{i=0}^j;\theta_{j}\big)\Delta W_{t_j}.
	\end{split}
\end{equation*}

\item Finally, the scheme is to optimize the global loss function
\begin{equation*}
	\widehat{\mathfrak{L}}(\cY_0,\cZ_0,\theta_{1},\theta_{2},\cdots,\theta_{N-1}):=\sum_{j=1}^ {N-1}\mathfrak{L}_j +\bE\bigg|Y_{t_{N}}^{\pi}-g(V_{t_N}X_{t_{N}}^{\pi})\bigg|^2,
\end{equation*}
over all $\theta=(\cY_0,\cZ_0,\theta_{1},\theta_{2},\cdots,\theta_{N-1})$, and if 
%$\theta^*=(\cY^*_0,\cZ^*_0,\theta^*_{1,1},\theta^*_{2,1},\cdots,\theta^*_{1,J},\theta^*_{2,J})$ if
\begin{equation*}
	\theta^*=(\cY^*_0,\cZ^*_0,\theta^*_{1},\theta^*_{2},\cdots,\theta^*_{N-1})\in \arg \min_{\theta\in\Theta} \widehat{\mathfrak{L}}(\theta),
\end{equation*}
then $\cY_0^*$ is the desired approximation of $Y^{\pi}_{t_0}$ by this method. 
\end{itemize}
In both algorithms, the input dimension of the neural network varies with time, while the output dimension remains constant across all time steps. Specifically, at time step $j$, the neural network receives an input of dimension $d+(j+1)(m_0+m)$, reflecting the inclusion of the state variable $X_t$, and the historical path information $\{W_s,V_s: 0\leq s\leq t_j\}$ encoded over $j+1$ time steps. 

For Algorithm 1, the output dimension of the neural network is $md_0$, and for Algorithm 2, it is $(1+m)d_0$. In both cases, the output dimension remains unchanged across time steps and corresponds to the approximated values of the backward processes at that step.

%%%%%%%%%%%%%%%%%%%%%%%%%%%%%%%%%%%%%%%%%%%%%%%%%%%%%%%%%%%%%%%%%%%%%%%%%%%%%%%%%%%%%%%%%%%%%%%%%%%%%%%%%%%%%%%%%%%%%%%%%%%%%%%%%%%%%%%%%%%%%%%%%%%%%%%%%%%%%%%%%%%%%%%%%

\section{Numerical Examples with Applications}

For the stochastic variance $V_t$, we consider rough Bergomi model (see~\cite{CBayerPFrizJHGath2016PricingUnderRoughVolat}), given by
\[ V_t=\xi_t\cE(\eta \widehat{W}_t),  \]
where $\xi_t$ is the forward variance curve (typically computed from the implied volatility surface), $\cE$ denotes the Wick exponential, defined by
\[\cE(Z)=\exp(Z-\frac{1}{2}\var(Z)),\] for a zero-mean normal random variable $Z$, and $\eta\geq0$ is a model parameter. The process $\widehat{W}$ is a fractional Brownian motion of Riemann-Liouville type with Hurst index $H\in (0,1/2)$, defined by

\begin{equation}
	\label{eq:RL-fbm}
	\widehat{W}_t \coloneqq \int_0^t \mathcal K(t-s) \, d\widetilde W_s, \quad \mathcal K(r) \coloneqq
	\sqrt{2H} r^{H-1/2}, \quad r > 0,
\end{equation}
where $\widetilde{W}$ is a standard Brownian motion.

\subsection{Example 1}
Our first example is based on Example 2 from Han and Long \cite{JHanJLong2020ConvgDeepBSDEforCoupledFBSDEs}, where we consider the following coupled FBSDE:
\begin{align}
\begin{cases}
    X_t &= X_0 + \int_{0}^{t}\bigg[ \mu \bigg(\sigma \rho Y_s e^{-r(T-s)} \lambda \cos\Big(X_s + \int_{0}^{s} \sin{(V_r)} dr\Big) - Z_s\bigg)\bigg]ds \\
    &\quad+ \int_{0}^{t} \sigma Y_s\Big(\rho dW_s + \sqrt{1-\rho^2}dB_s\Big) ,\\
    X_0 &= \frac{\pi}{2} ,\\
    -dY_t &= \bigg[-rY_t + \frac{1}{2}e^{-3r(T-t)}\sigma^2 \lambda^3 \sin^3\Big(X_t + \int_{0}^{t} \sin(V_s)ds\Big)\\
    &\quad- \sin(V_t)e^{-r(T-t)}\lambda \cos\Big(X_t + \int_{0}^{t} \sin(V_s)ds\Big)\bigg]dt - Z_tdW_t - \Tilde{Z_t}dB_t, \\
    Y_T &= \lambda \sin\Big( X_T + \int_{0}^{T} \sin{(V_s)}ds \Big).
\end{cases}
\end{align}
The analytical solution of the above coupled  FBSDE is known to be: $$Y_t = e^{-r(T-t)}\lambda \sin{\Big(X_t + \int_{0}^{t} \sin{(V_s)}ds\Big)}.$$ 
In particular, at $t=0$ and $X_0 = \frac{\pi}{2}$, this yields the exact value:
 $$Y_0 \approx 7.6098354.$$

For numerical simulations, we set the parameters as follows: $T=1$, $\rho = -0.9$, $r = 0.05$, $\lambda = 8$, $\mu = 0.1$, $\sigma = 0.2$, $H = 0.07$, $\eta=1.9$, and $\xi_t = 0.09$. We apply both of our algorithms using various time discretizations with $N\in\{5,10,20,30\}$.

We train the model using the Adam optimizer on 64000 samples for 20 epochs.  Performance is evaluated after each epoch on a separate test set of 16,000 samples. A batch size of 64 is used for both training and testing, implying 1,000 training updates per epoch. The neural networks consist of 3 hidden layers with ReLU activation and include batch normalization before the hidden layers. The neural network approximation at time step $j$ is denoted by
$$\cX_j^{\cN}(\cdot;\theta_j)=\big(\cY_j^{\mathcal{N}}(\cdot;\theta_{j}), \cZ_j^{\mathcal{N}}(\cdot;\theta_{j})\big)$$
where the network input is taken as $\big(X_{t_j}^{\pi},\{V_{t_i}\}_{i=0}^j\big)$.

\textbf{Algorithm 1 Results: }
Table~\ref{tab:Alg_1Ex_1} shows the results from Algorithm 1, reporting the mean approximation of $Y_0$, relative error with respect to the true solution, and the empirical variance across 20 independent runs. The approximation improves with increasing $N$, demonstrating the algorithm’s effectiveness and consistency.
\begin{table}[ht]
\caption{Value of $Y_0$ from Algorithm 1 (Example 1)} % title of Table
\centering 
\begin{tabular}{c c c c}
\hline\hline 
N & Approximation & Relative Error & Variance  \\ [0.5ex] % inserts table
%heading
\hline % inserts single horizontal line
5 & 7.79166896 & 2.389 \% &  0.000098\\ % inserting body of the table
10 & 7.67836123 & 0.900\% & 0.000054\\
20 & 7.64265652 & 0.431\% & 0.000063 \\
30 & 7.63218625 & 0.294\% & 0.000069\\ [1ex] % [1ex] adds vertical space
\hline %inserts single line
\end{tabular}
\label{tab:Alg_1Ex_1} % is used to refer this table i
\end{table}

\textbf{Algorithm 2 Results: }Table~\ref{tab:Alg_2Ex_1} presents results from Algorithm 2 under the same setup. Although the variance remains low and the estimates are consistent, the accuracy is slightly inferior to Algorithm 1, with a persistent relative error around 1.8–2.0\%.

\begin{table}[ht]
\caption{Value of $Y_0$ from Algorithm 2 (Example 1)} % title of Table
\centering % used for centering table
\begin{tabular}{c c c c} % centered columns (4 columns)
\hline\hline %inserts double horizontal lines
N & Approximation & Relative Error & Variance\\ [0.5ex] % inserts table
%heading
\hline % inserts single horizontal line
5 & 7.76642227 & 2.058\% &  0.000091\\ % inserting body of the table
10 & 7.75322528 & 1.884\%  &  0.000146\\
20 & 7.76028776 & 1.977\%  &  0.000070 \\
30 & 7.74844878 & 1.822\%  &  0.000080\\ [1ex] % [1ex] adds vertical space
\hline %inserts single line
\end{tabular}
\label{tab:Alg_2Ex_1} % is used to refer this table i
\end{table}

%\newpage

\subsection{Example 2: Utility Maximization with Fixed Consumption Rate}

%\subsubsection{}
In this example, we consider a utility maximization problem for an agent (or investor) who allocates their wealth between risky and risk-free assets, with a consumption rate influenced by an external process. The risky asset $S_t$ follows a stochastic volatility model driven by a two-dimensional Weiner process $W =  (\widetilde{W}, B)$:
\[
\begin{cases}
	dS_t &= (r + \theta V_t)S_tdt + S_t \sqrt{V_t}(\rho d\widetilde{W}_t +\sqrt{1-\rho^2}dB_t), \\
	S_0 &= s_0,
\end{cases}\]
where $\rho \in [-1,1]$ is the correlation coefficient, $r>0$ is the risk-free interest rate, and $\theta \sqrt{V_t};\theta\neq 0$ is the market price of risk (risk premium).

Let $\pi_t$ denote the proportion of the wealth invested in the risky asset, and define $u_t := \sqrt{V_t}\pi_t$. We assume the consumption rate is given by $c_t\sqrt{V_t}X_t$, where $c_t\ge 0$ is a bounded deterministic function representing the agent's personal consumption preference. Then, the wealth process $X_t$ evolves according to the SDE:
\begin{align}
	dX_t=\left(
	r+\theta \sqrt{V_t}u_t-c_t\sqrt{V_t}\right)X_t\,dt
	+u_tX_t\,\left( \rho \,d\widetilde W_t + \sqrt{1-\rho^2} \,dB_t  \right), \quad X_0=x_0>0.
	\label{wealth-SDE}
\end{align}

The agent aims to maximize their expected utility over the interval $[0,T]$, with the utility functional:
\begin{equation}\label{utlity_max_power_1}
	U(0,x_0)=\max_{u} \bE\bigg[ e^{-rT}X^{\gamma}_T+\int_0^T F(t,X_t,u_t)dt\bigg],\quad X_0=x_0,\; 0<\gamma<1,
\end{equation}
where the running utility function $F$ is given by:
\begin{align}
	F(t,X_t,u_t)&=e^{-r t}\bigg(a_t\sqrt{V_t}\Big(1-\pi_t\Big)X_t+\Big(c_t\sqrt{V_t}X_t\Big)^{\gamma}\bigg)\label{running_cost_1}\\
	&=e^{-r t}\bigg(a_t\Big(\sqrt{V_t}-u_t\Big)X_t+\Big(c_t\sqrt{V_t}X_t\Big)^{\gamma}\bigg)\label{running_cost}.
\end{align}
Here $a_t\ge 0$ rewards deviations from investment in the risky asset. When $a_t=0$, the problem reduces to the classical power utility maximization with consumption.

By It\^o-Doeblin formula, we obtain
\begin{equation}
	\bE \bigg[ X^{\gamma}_T\bigg]= x^{\gamma}_0+\bE \int_0^T \gamma X^{\gamma}_t \bigg(r+\theta \sqrt{V_t}u_t-c_t\sqrt{V_t}+\frac{\gamma-1}{2} u_t^2 \bigg)dt.
\end{equation}

Thus, the agent's utility maximization problem in (\ref{utlity_max_power_1}) is equivalent to the following optimization problem:
\begin{equation}
	\max_u \bE\bigg[ \int_0^T \bigg( e^{-rT}\gamma X^{\gamma}_t \Big(r+\theta \sqrt{V_t}u_t-c_t\sqrt{V_t}+\frac{\gamma-1}{2} u_t^2 \Big)+e^{-r t}\bigg(a_t\Big(\sqrt{V_t}-u_t\Big)X_t+\Big(c_t\sqrt{V_t}X_t\Big)^{\gamma}\bigg)dt\bigg].
\end{equation}
Let $u_t^*$ (or equivalently $\pi^*_t=\frac{u^*_t}{\sqrt{V_t}}$) denote the optimal control. Applying Pontryagin's maximum principle, we arrive at the following coupled FBSDE system for $t\in[0,T]$:

\begin{equation}\label{utlity_max_FBSDE}
	\begin{cases}
		dX_t &= (r + \theta \sqrt{V_t}u_t^*-c_t\sqrt{V_t})X_tdt + u_t^*X_t(\rho d\widetilde{W}_t + \sqrt{1-\rho^2}dB_t, \\
		X_0 &= x_0, \\
		-dY_t &= \Bigg(\big(r+ \theta \sqrt{V_t}u_t^*-c_t\sqrt{V_t}\big)Y_t + u_t^*\big(\rho \widetilde{Z}_t + \sqrt{1- \rho^2}Z_t\big)\\
		&\quad+ e^{-rT}\gamma^2 X_t^{\gamma-1}\Big(r_t + \theta \sqrt{V_t}u_t^*-c_t\sqrt{V_t}+ \frac{\gamma-1}{2} \abs{u_t^*}^2\Big)\\
		&\quad+e^{-r t}\Big(a_t(\sqrt{V_t}-u_t)+\Big(c_t\sqrt{V_t}\Big)^{\gamma}\gamma X_t^{\gamma-1}\Big)\Bigg)dt - Z_tdB_t - \widetilde{Z}_td\widetilde{W}_t,\\
		Y_T &= 0 ,\\
		u_t^* &= \frac{ \theta \sqrt{V_t}(X_t^{1-\gamma}Y_t+e^{-rT}\gamma)+X_t^{1-\gamma}\big(\sqrt{1- \rho^2}Z_t + \rho \widetilde{Z}_t\big)-e^{-r t}a_tX_t^{1-\gamma} }{\gamma(1-\gamma)e^{-rT}}.
	\end{cases}
\end{equation}

%\newpage
To solve for the optimal strategy, optimal wealth, and the optimal expected utility, we must numerically solve the coupled FBSDE in (\ref{utlity_max_FBSDE}). For this purpose, we set $c=0.02$ and use the same parameter values as in Example 1, namely, $T=1$, $r = 0.05$, $\rho=-0.9$, $H = 0.07$, $\eta=1.9$, and $\xi_t = 0.09$. We then employ Algorithm 1 to compute the optimal strategy $\pi^*_t$, the optimal wealth process $X_t$, and consequently the optimal expected utility.

To improve generalization and avoid overfitting, we reduce the complexity of our neural network by using only two hidden layers and applying $l_1$-regularization. Additionally, we incorporate gradient clipping to prevent exploding gradients. The model is trained over 30 epochs.

Figure \ref{fig:Ex2_loss} illustrates the average training and testing loss over epochs for $N=100$ with parameters $\gamma=0.8,\theta=0.3,a=0.025,x_0=1$, and $s_0=1$. The results indicate that the neural network performs consistently well on both training and testing data.

\begin{figure}[h]
	\centering
	\includegraphics[width=0.7\linewidth]{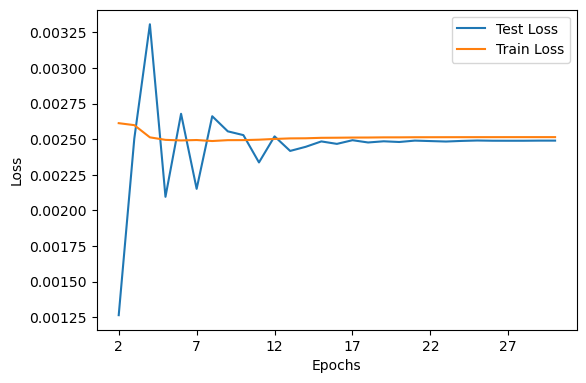}
	\caption{Train and Test Loss During Epochs ($\gamma=0.8, \theta=0.3,a=0.025$)}
	\label{fig:Ex2_loss}
\end{figure}

Figure \ref{fig:sample_paths_ex2N_gam_08_theta_03_a_0}-\ref{fig:sample_paths_ex2N_gam_08_theta_01_a_0025} presents sample paths of the processes $S_t,\pi_t^*$ and $X_t$ under different choices of parameters $\gamma,\theta$ and $a$, while all cases share the same rough volatility process $V_t$. Finally, Figure \ref{fig:consumption rates} displays the corresponding consumption rates for all four parameter configurations.

\begin{figure}
	\begin{subfigure}[b]{0.43\textwidth}
		\includegraphics[width=1.1\linewidth]{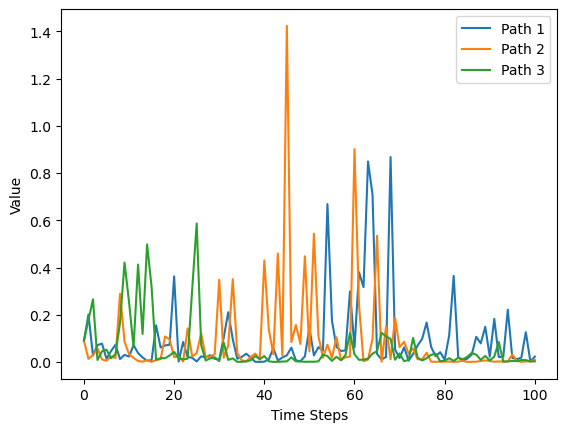}
		\caption{Rough Volatility Process $V_t$}
	\end{subfigure}
	\hfill
	\begin{subfigure}[b]{0.43\textwidth}
		\includegraphics[width=1.1\linewidth]{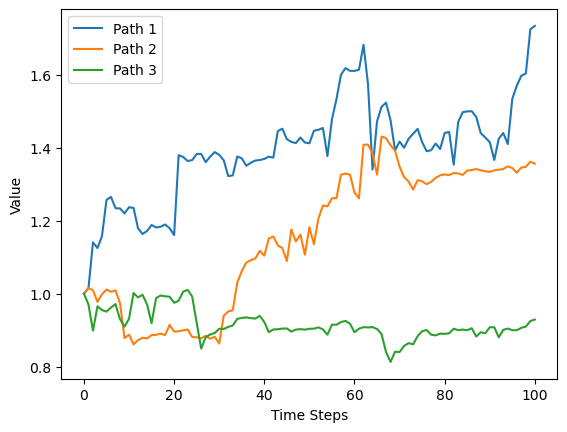}
		\caption{Stock Price Process $S_t$}
	\end{subfigure}
	\begin{subfigure}[b]{0.43\textwidth}
		\includegraphics[width=1.1\linewidth]{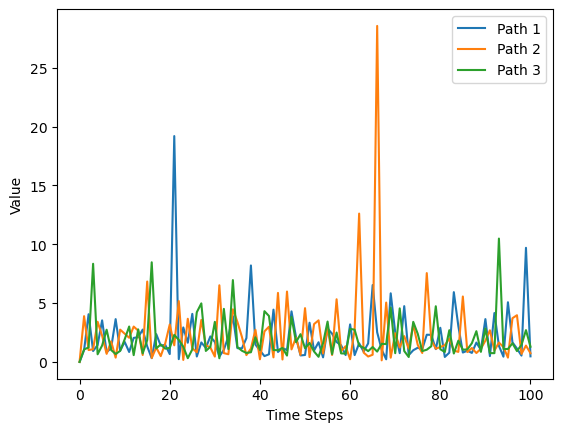}
		\caption{Optimal Strategy $\pi^*_t$}
	\end{subfigure}
	\hfill
	\begin{subfigure}[b]{0.43\textwidth}
		\includegraphics[width=1.1\linewidth]{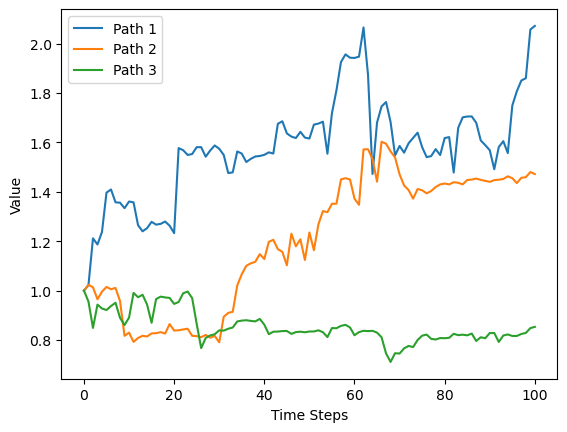}
		\caption{Optimal Wealth $X_t$}
	\end{subfigure}
	
	\caption{Sample Paths with $\gamma=0.8, \theta=0.3,a=0.0$}
	\label{fig:sample_paths_ex2N_gam_08_theta_03_a_0}
\end{figure}

\begin{figure}
	\begin{subfigure}[b]{0.43\textwidth}
		\includegraphics[width=1.1\linewidth]{ex2N_V_N_100_gam_08_theta_03_a_0025_c002_X0_1.png}
		\caption{Rough Volatility Process $V_t$}
	\end{subfigure}
	\hfill
	\begin{subfigure}[b]{0.43\textwidth}
		\includegraphics[width=1.1\linewidth]{ex2N_S_N_100_gam_08_theta_03_a_0025_c002_X0_1.png}
		\caption{Stock Price Process $S_t$}
	\end{subfigure}
	\begin{subfigure}[b]{0.43\textwidth}
		\includegraphics[width=1.1\linewidth]{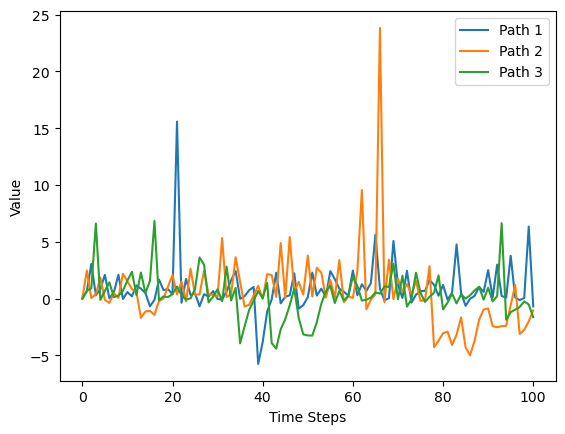}
		\caption{Optimal Strategy $\pi^*_t$}
	\end{subfigure}
	\hfill
	\begin{subfigure}[b]{0.43\textwidth}
		\includegraphics[width=1.1\linewidth]{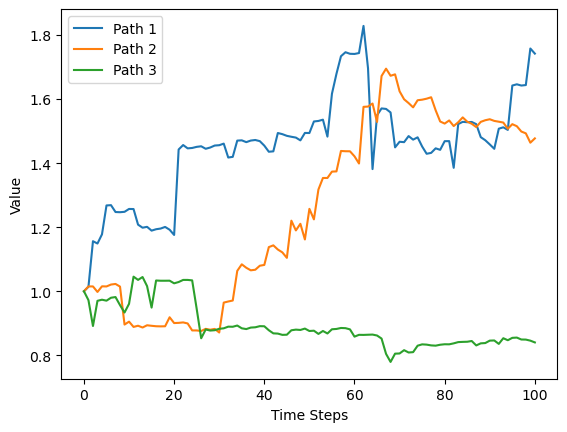}
		\caption{Optimal Wealth $X_t$}
	\end{subfigure}
	
	\caption{Sample Paths with $\gamma=0.8, \theta=0.3,a=0.025$}
	\label{fig:sample_paths_ex2N_gam_08_theta_03_a_0025}
\end{figure}

\begin{figure}
	\begin{subfigure}[b]{0.43\textwidth}
		\includegraphics[width=1.1\linewidth]{ex2N_V_N_100_gam_08_theta_03_a_0025_c002_X0_1.png}
		\caption{Rough Volatility Process $V_t$}
	\end{subfigure}
	\hfill
	\begin{subfigure}[b]{0.43\textwidth}
		\includegraphics[width=1.1\linewidth]{ex2N_S_N_100_gam_08_theta_03_a_0025_c002_X0_1.png}
		\caption{Stock Price Process $S_t$}
	\end{subfigure}
	\begin{subfigure}[b]{0.43\textwidth}
		\includegraphics[width=1.1\linewidth]{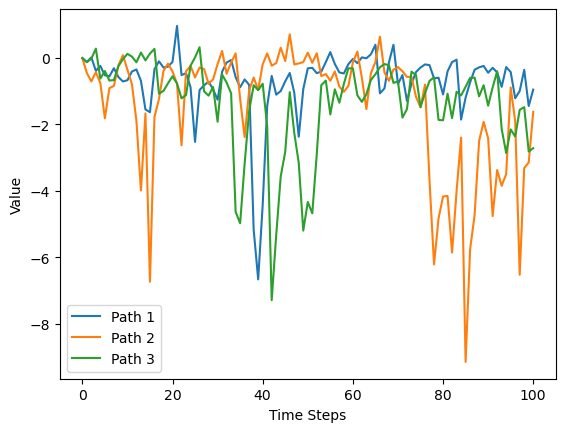}
		\caption{Optimal Strategy $\pi^*_t$}
	\end{subfigure}
	\hfill
	\begin{subfigure}[b]{0.43\textwidth}
		\includegraphics[width=1.1\linewidth]{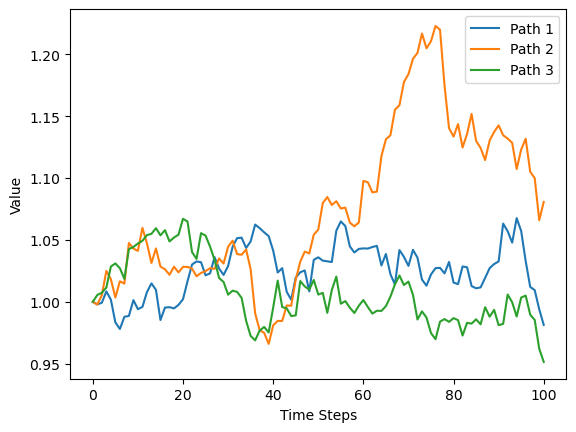}
		\caption{Optimal Wealth $X_t$}
	\end{subfigure}
	
	\caption{Sample Paths with $\gamma=0.2, \theta=0.3,a=0.025$}
		\label{fig:sample_paths_ex2N_gam_02_theta_03_a_0025}
\end{figure}

\begin{figure}
	\begin{subfigure}[b]{0.43\textwidth}
		\includegraphics[width=1.1\linewidth]{ex2N_V_N_100_gam_08_theta_03_a_0025_c002_X0_1.png}
		\caption{Rough Volatility Process $V_t$}
	\end{subfigure}
	\hfill
	\begin{subfigure}[b]{0.43\textwidth}
		\includegraphics[width=1.1\linewidth]{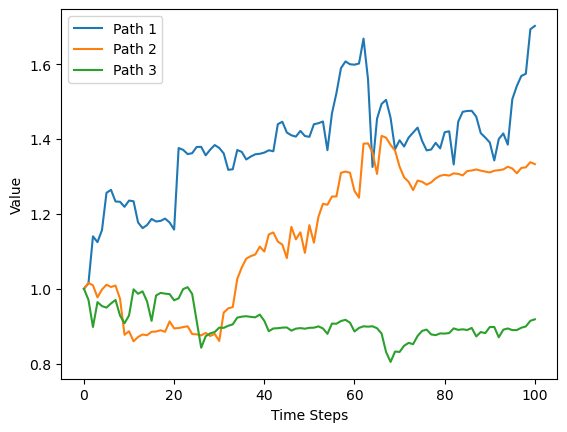}
		\caption{Stock Price Process $S_t$}
	\end{subfigure}
	\begin{subfigure}[b]{0.43\textwidth}
		\includegraphics[width=1.1\linewidth]{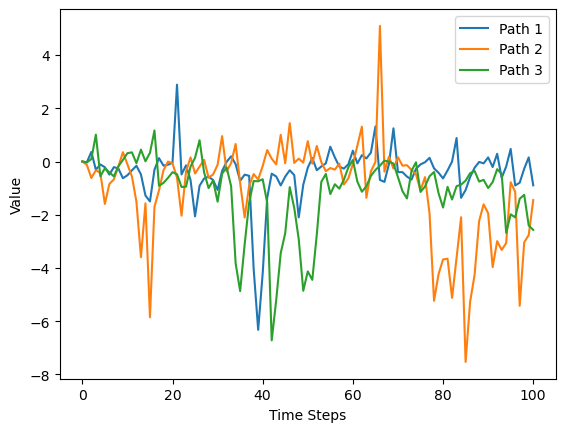}
		\caption{Optimal Strategy $\pi^*_t$}
	\end{subfigure}
	\hfill
	\begin{subfigure}[b]{0.43\textwidth}
		\includegraphics[width=1.1\linewidth]{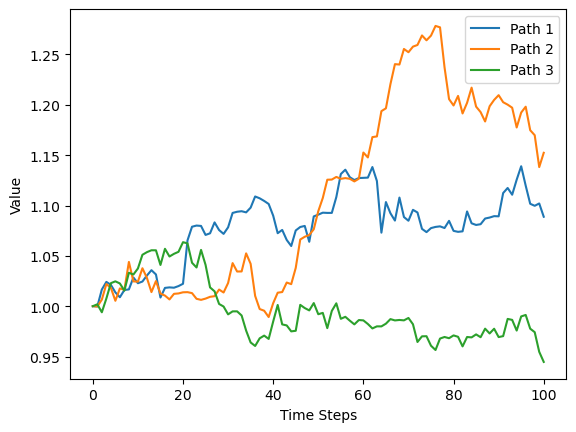}
		\caption{Optimal Wealth $X_t$}
	\end{subfigure}
	
	\caption{Sample Paths with $\gamma=0.8, \theta=0.1,a=0.025$}
		\label{fig:sample_paths_ex2N_gam_08_theta_01_a_0025}
\end{figure}

%\begin{figure}[h]
%	\centering
%	\includegraphics[width=0.7\linewidth]{ex2N_C_N_100_gam_08_theta_01_a_0025_c002_X0_1.png}
%	\caption{Consumption Rate ($\gamma=0.8, \theta=0.1,a=0.025$)}
%	\label{fig:Ex2_consumption}
%\end{figure}

\begin{figure}
	\begin{subfigure}[b]{0.43\textwidth}
		\includegraphics[width=1.1\linewidth]{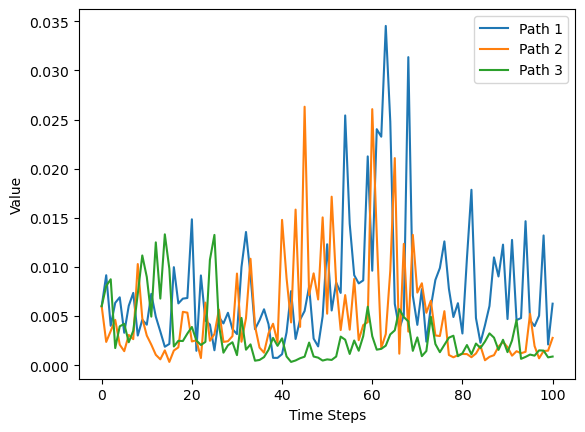}
		\caption{$\gamma=0.8, \theta=0.3,a=0.0$}
	\end{subfigure}
	\hfill
	\begin{subfigure}[b]{0.43\textwidth}
		\includegraphics[width=1.1\linewidth]{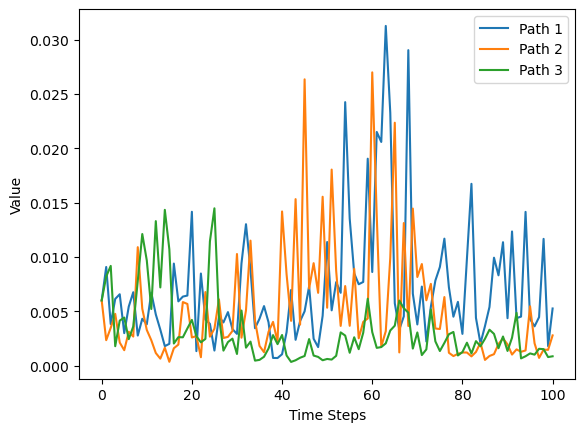}
		\caption{$\gamma=0.8, \theta=0.3,a=0.025$}
	\end{subfigure}
	\begin{subfigure}[b]{0.43\textwidth}
		\includegraphics[width=1.1\linewidth]{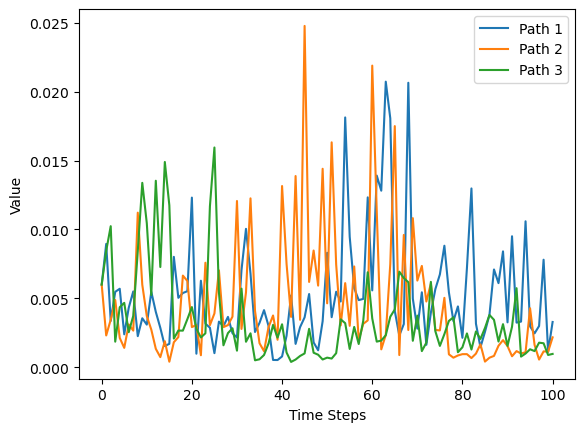}
		\caption{$\gamma=0.8, \theta=0.1,a=0.025$}
	\end{subfigure}
	\hfill
	\begin{subfigure}[b]{0.43\textwidth}
		\includegraphics[width=1.1\linewidth]{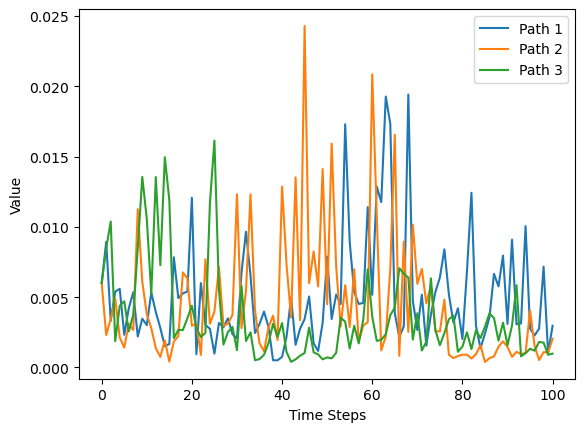}
		\caption{$\gamma=0.2, \theta=0.3,a=0.025$}
	\end{subfigure}
	
	\caption{Sample Paths of Consumption Rate}
	\label{fig:consumption rates}
\end{figure}

\section{Conclusion}
We developed a deep learning-based framework for fully coupled non-Markovian FBSDEs and established error bounds and convergence results for this class of equations. Our analysis extends the scope of Deep BSDE methods by allowing the coefficients of the forward process to be random and depending on all the backward components. The effectiveness of the scheme was demonstrated through utility maximization problems under rough volatility, illustrating its accuracy and applicability to complex, non-Markovian models.

%\nocite{*}
\addcontentsline{toc}{section}{References}
\bibliography{2021summerref}

@article{JMaHYinJZhang2012OnNon-MarkovFBSDEsandBSPDEs,
title = {On non-Markovian forward–backward SDEs and backward stochastic PDEs},
journal = {Stochastic Processes and their Applications},
volume = {122},
number = {12},
pages = {3980-4004},
year = {2012},
issn = {0304-4149},
doi = {https://doi.org/10.1016/j.spa.2012.08.002},
url = {https://www.sciencedirect.com/science/article/pii/S030441491200169X},
author = {Ma, Jin and Yin, Hong and Zhang, Jianfeng}
}

@article{SaporitoZhang2021NeuralNetPathDependentPDE,
	title={PDGM: a Neural Network Approach to Solve Path-Dependent Partial Differential Equations},
	author={Saporito, Y.  and Zhang, Z.},
	journal={SIAM Journal on Financial Mathematics},
	year={2021},
	volume={12},
	number={3},
	pages={912-940}
}

@article{QiFeng_2023_DeepSignatureFBSDEAlgo,
	title={Deep Signature FBSDE Algorithm},
	author={Feng, Qi and Luo, Man and Zhang, Zhaoyu},
	journal={Numerical Algebra, Control and Optimization},
	year={2023},
	volume={13},
	number={3\&4},
	pages={500-522}
}

@article{TPhamJZhang_2014_ZeroSumGamePathDependent,
	title={Two Person Zero-Sum Game in Weak Formulation and Path Dependent Bellman-Isaacs Equation},
	author={Pham, T. and Zhang, J.},
	journal={SIAM Journal on Control and Optimization},
	year={2014},
	volume={52},
	number={4},
	pages={2090-2121}
}

@article{Saporito2019StochControlPathDependent,
	title={Stochastic Control and Differential Games with Path-Dependent Influence of Controls on Dynamics and Running Cost},
	author={Saporito, Yuri F.},
	journal={SIAM Journal on Control and Optimization},
	year={2019},
	volume={57},
	number={2},
	pages={1312-1327}
}

@article{FAntonelli1993BFSDEs,
	title={Backward-forward stochastic differential equations},
	author={Antonelli, Fabio},
	journal={The Annals of Applied Probability},
	year={1993},
	volume={3},
	number={3},
	pages={777-793}
}

@article{YHuSPeng1995SolFBSDEs,
	title={Solution of forward-backward stochastic differential equations},
	author={Hu, Y. and Peng, S.},
	journal={Probability Theory and Related Fields},
	year={1995},
	volume={103},
	pages={273-283}
}

@article{JYong1997FindAdaptedSolFBSDEMethodofCont,
	title={Finding adapted solutions of forward-backward stochastic differential equations: method of continuation},
	author={Yong, Jiongmin},
	journal={Probability Theory and Related Fields},
	year={1997},
	volume={107},
	pages={537-572}
}

@article{JMaPProtterJYong1994SolvingFBSDEs4StepScheme,
  title={Solving forward-backward stochastic differential equations explicitly — a four step scheme},
  author={Ma, Jin and Protter, Philip  and Yong, Jiongmin},
  journal={Probability Theory and Related Fields},
  year={1994},
  volume={98},
  pages={339-359}
}

@article{JMaZWuDZhangJZhang2015OnWellPossedofFBSDEsAUnifiedAppr,
   title={On well-posedness of forward–backward SDEs—A unified approach},
   volume={25},
   ISSN={1050-5164},
   url={http://dx.doi.org/10.1214/14-AAP1046},
   DOI={10.1214/14-aap1046},
   number={4},
   journal={The Annals of Applied Probability},
   publisher={Institute of Mathematical Statistics},
   author={Ma, Jin and Wu, Zhen and Zhang, Detao and Zhang, Jianfeng},
   year={2015},
   month={Aug}
}

@article {JHanAJentWE2018SolvingHighDimEqusDeepLearnin,
	author = {Han, Jiequn and Jentzen, Arnulf and E, Weinan},
	title = {Solving high-dimensional partial differential equations using deep learning},
	volume = {115},
	number = {34},
	pages = {8505--8510},
	year = {2018},
	publisher = {National Academy of Sciences}
}

@article{EWeinanJHanAJent2017DeepLearnNumrMethodHighDimParbPDEsBSDEs,
   title={Deep Learning-Based Numerical Methods for High-Dimensional Parabolic Partial Differential Equations and Backward Stochastic Differential Equations},
   volume={5},
   number={4},
   journal={Communications in Mathematics and Statistics},
   author={E., Weinan and Han, Jiequn and Jentzen, Arnulf},
   year={2017},
   pages={349–380}
}

@article{CHureHPhamXWarin2020DeepBShcemeHighDimNonlPDEs,
      title={Deep backward schemes for high-dimensional nonlinear PDEs}, 
      author={Huré, Côme  and Pham, Huyên  and Warin, Xavier},
      year={2020},
      eprint={1902.01599},
      archivePrefix={arXiv},
      primaryClass={math.PR}
}

@article{SJiSPengYPengXZhang2020ThreeAlgorithmsforSolvHighDimCoupledFBSDEsDeepLearn,
      title={Three algorithms for solving high-dimensional fully-coupled FBSDEs through deep learning}, 
      author={Ji, Shaolin and Peng, Shige  and Peng, Ying and Zhang, Xichuan},
      year={2020},
    journal={IEEE Intelligent Systems},
    volume = {35},
    number = {3},
    pages={71-84}
}

@article{CBayerJQiuYYao2020PricingOptnRoughVolatBSPDEs,
      title={Pricing Options Under Rough Volatility with Backward SPDEs}, 
      author={Bayer, Christian and Qiu, Jinniao and Yao, Yao},
      year={2020},
      eprint={2008.01241},
      archivePrefix={arXiv},
      primaryClass={q-fin.MF}
}

@article{CBayerPFrizJHGath2016PricingUnderRoughVolat,
  title={Pricing under rough volatility},
  author={Bayer, Christian and Friz, Peter and Gatheral, Jim},
  journal={Quantitative Finance},
  volume={16},
  number={6},
  pages={887--904},
  year={2016},
  publisher={Taylor \& Francis}
}

@article{CBenderJZhang2008TimeDiscrMarkovIteraCoupledDFBSDEs,
   title={Time discretization and Markovian iteration for coupled FBSDEs},
   volume={18},
   number={1},
   journal={The Annals of Applied Probability},
   author={Bender, Christian and Zhang, Jianfeng},
   year={2008}
}

@article{JHanJLong2020ConvgDeepBSDEforCoupledFBSDEs,
   title={Convergence of the deep BSDE method for coupled FBSDEs},
   volume={5},
   ISSN={2367-0126},
   url={http://dx.doi.org/10.1186/s41546-020-00047-w},
   DOI={10.1186/s41546-020-00047-w},
   number={1},
   journal={Probability, Uncertainty and Quantitative Risk},
   publisher={American Institute of Mathematical Sciences (AIMS)},
   author={Han, Jiequn and Long, Jihao},
   year={2020},
   month={July}
}

@article{BBouchardNTouzi2004DiscreteTimeApprxMCsimulBSDE,
	title={Discrete-time approximation and Monte-Carlo simulation of backward stochastic differential equations},
	author={Bouchard, Bruno and Touzi, Nizar},
	journal={Stochastic Process and their Applications},
	year={2004},
	volume={111},
	pages={175-206}
}

@article{EGobetJPLemorXWarin2005RegressionBasedMCtoSolveBSDEs,
	title={A Regression-Based Monte Carlo Method to Solve Backward Stochastic Differential Equations},
	author={Gobet, Emmanuel and Lemor, JeanPhilippe and Warin, Xavier},
	journal={The Annals of Applied Probability},
	year={2005},
	volume={15},
	number={3},
	pages={2172-2202}
}

@article{JZhang2004NumericalSchemeBSDEs,
	title={A Numerical Scheme for BSDEs},
	author={Zhang, Jianfeng},
	journal={The Annals of Applied Probability},
	year={2004},
	volume={14},
	number={1},
	pages={459-488}
}

@Article{HUMollaJQiu20221NumercFBSDE,
	author = {Molla, Hasib Uddin and Qiu, Jinniao},
	title = {Numerical Approximations of coupled forward-backward SPDEs},
	journal = {Stochastic analysis and applications},
	year={2021},
	url={https://doi.org/10.1080/07362994.2021.2011318}
}
\bibliographystyle{alpha}

\appendix
\section{Proof of Lemma \ref{R3_fw_est}}
\begin{proof}
	For $\phi=b,\sigma$, define the differences:
	\begin{equation}
		\Delta\phi:=\phi(t_i,V_{t_i},X_{t_i}^1,Y_{t_i}^1,Z_{t_i}^1)-\phi(t_i,V_{t_i},X_{t_i}^2,Y_{t_i}^2,Z_{t_i}^2).
	\end{equation}
	Then the forward process difference satisfies:
	\begin{equation}
		\Delta X_{t_{i+1}}=\Delta X_{t_i}+\Delta bh+\int_{t_i}^{t_{i+1}}\Delta\alpha_t dt+\Delta\sigma\Delta W_{{t_i}}+\int_{t_i}^{t_{i+1}}\Delta\beta_t dW_t.
	\end{equation}
Rearranging terms and squaring on both sides yields that
	\begin{equation*}
	|\Delta X_{t_{i+1}}-\Delta X_{t_i}|^2=\bigg|\Delta bh+\Delta\sigma\Delta W_{t_i}+\int_{t_i}^{t_{i+1}}\Delta\alpha_t dt+\int_{t_i}^{t_{i+1}}\Delta\beta_t dW_t\bigg|^2.
\end{equation*}
Hence, expanding $|\Delta X_{t_{i+1}}|^2$:
	\begin{align*}
	&|\Delta X_{t_{i+1}}|^2 \\
	&=-|\Delta X_{t_i}|^2+2\Delta X_{t_i} \Delta X_{t_{i+1}}+\bigg|\Delta bh+\Delta\sigma\Delta W_{t_i}+\int_{t_i}^{t_{i+1}}\Delta\alpha_t dt+\int_{t_i}^{t_{i+1}}\Delta\beta_t dW_t\bigg|^2 \\
	&=-|\Delta X_{t_i}|^2+2\Delta X_{t_i} \bigg(\Delta X_{t_i}+\Delta bh+\Delta\sigma\Delta W_{{t_i}}+\int_{t_i}^{t_{i+1}}\Delta\alpha_t dt+\int_{t_i}^{t_{i+1}}\Delta\beta_t dW_t\bigg)\\
	& \quad +|\Delta b|^2 h^2 +|\Delta \sigma|^2 |\Delta W_{t_i}|^2 + \bigg| \int_{t_i}^{t_{i+1}}\Delta\alpha_t dt\bigg|^2 + \bigg| \int_{t_i}^{t_{i+1}}\Delta\beta_t dW_t\bigg|^2 \\
	& \quad +2h \Delta b\Delta \sigma \Delta W_{t_i}+ 2h \Delta b \int_{t_i}^{t_{i+1}}\Delta\alpha_t dt + 2h\Delta b \int_{t_i}^{t_{i+1}}\Delta\beta_t dW_t +2\Delta\sigma \Delta W_{t_i} \int_{t_i}^{t_{i+1}}\Delta\alpha_t dt\\
	& \quad+ 2 \Delta \sigma\Delta W_{t_i} \int_{t_i}^{t_{i+1}}\Delta\beta_t dW_t+ 2 \int_{t_i}^{t_{i+1}}\Delta\alpha_t dt \int_{t_i}^{t_{i+1}}\Delta\beta_t dW_t. \\
\end{align*}
Taking conditional expectation $\bE_{t_i}[\cdot]=\bE[\cdot|\cF_{t_i}]$ on both sides and using the fact that $X_{t_i}^l,Y_{t_i}^l,Z_{t_i}^l$ are $\mathcal{F}_{t_i}$-measurable and so are $\Delta b, \Delta\sigma$, gives

\begin{align*}
	&\bE_{t_i}|\Delta X_{t_{i+1}}|^2 \\
	 &\le|\Delta X_{t_i}|^2+2\Delta X_{t_i} \bigg(\Delta bh+\bE_{t_i}\int_{t_i}^{t_{i+1}}\Delta\alpha_t dt\bigg)\\
	 & \quad +|\Delta b|^2 h^2 +|\Delta \sigma|^2 h +h \bE_{t_i} \int_{t_i}^{t_{i+1}}|\Delta\alpha_t |^2dt + \bE_{t_i} \int_{t_i}^{t_{i+1}}|\Delta\beta_t|^2 dt \\
	 & \quad + 2h \Delta b \bE_{t_i}\int_{t_i}^{t_{i+1}}\Delta\alpha_t dt + 2\bE_{t_i}\bigg[\Delta\sigma \Delta W_{t_i} \int_{t_i}^{t_{i+1}}\Delta\alpha_t dt\bigg]\\
	 & \quad+ 2\bE_{t_i} \bigg[ \Delta \sigma\Delta W_{t_i} \int_{t_i}^{t_{i+1}}\Delta\beta_t dW_t\bigg]+ 2 \bE_{t_i}\bigg[\int_{t_i}^{t_{i+1}}\Delta\alpha_t dt \int_{t_i}^{t_{i+1}}\Delta\beta_t dW_t\bigg] \\
	&\le |\Delta X_{t_i}|^2+\frac{\lambda_1}{2} h |\Delta X_{t_i}|^2 +\frac{2h}{\lambda_1} |\Delta b|^2+\frac{\lambda_1}{2} h|\Delta X_{t_i}|^2 +  \frac{2}{\lambda_1 h}\bigg|\bE_{t_i}\int_{t_i}^{t_{i+1}}\Delta\alpha_t dt\bigg|^2\\
	& \quad +|\Delta b|^2 h^2 +|\Delta \sigma|^2h + h \bE_{t_i} \int_{t_i}^{t_{i+1}}|\Delta\alpha_t |^2dt + \bE_{t_i} \int_{t_i}^{t_{i+1}}|\Delta\beta_t|^2 dt \\
	& \quad + \lambda_2 h^2 |\Delta b|^2 +\frac{h}{\lambda_2} \bE_{t_i}\int_{t_i}^{t_{i+1}}|\Delta\alpha_t|^2 dt +\frac{\lambda_1}{2}\bE_{t_i} |\Delta \sigma \Delta W_{t_i}|^2 + \frac{2h}{\lambda_1}\bE_{t_i} \int_{t_i}^{t_{i+1}}|\Delta\alpha_t|^2 dt\\
	& \quad + \frac{\lambda_1 }{2}\bE_{t_i}|\Delta \sigma \Delta W_{t_i}|^2+ \frac{2}{\lambda_1} \bE_{t_i}\int_{t_i}^{t_{i+1}}|\Delta\beta_t|^2 dt + \lambda_2h\bE_{t_i} \int_{t_i}^{t_{i+1}}|\Delta\alpha_t|^2 dt +\frac{1}{\lambda_2} \bE_{t_i}\int_{t_i}^{t_{i+1}}|\Delta\beta_t|^2 dt\\
	&\le(1+\lambda_1h) |\Delta X_{t_i}|^2 +\bigg( \frac{2h}{\lambda_1}+h^2+\lambda_2h^2\bigg) |\Delta b|^2+(h+\lambda_1 h) |\Delta \sigma|^2\\
	& \quad   +\bigg(\frac{2}{\lambda_1 }+h+\frac{h}{\lambda_2}+\frac{2h}{\lambda_1}+\lambda_2 h\bigg)\bE_{t_i}\int_{t_i}^{t_{i+1}}|\Delta\alpha_t|^2 dt + \bigg( 1+\frac{2}{\lambda_1}+\frac{1}{\lambda_2}\bigg)\bE_{t_i}  \int_{t_i}^{t_{i+1}}|\Delta\beta_t|^2 dt \\
	&\le(1+\lambda_1h) |\Delta X_{t_i}|^2 +\bigg( \frac{2h}{\lambda_1}+h^2+\lambda_2h^2\bigg) \big(K|\Delta X_{t_i}|^2+b_y|\Delta Y_{t_i}|^2+b_z|\Delta Z_{t_i}|^2\big)\\
	&\quad +(h+\lambda_1 h) \bigg(\sigma_x|\Delta X_{t_i}|^2+\sigma_y|\Delta Y_{t_i}|^2+\sigma_z|\Delta Z_{t_i}|^2\bigg)\\
	& \quad   +\bigg(\frac{2}{\lambda_1 }+h+\frac{h}{\lambda_2}+\frac{2h}{\lambda_1}+\lambda_2 h\bigg)\bE_{t_i}\int_{t_i}^{t_{i+1}}|\Delta\alpha_t|^2 dt + \bigg( 1+\frac{2}{\lambda_1}+\frac{1}{\lambda_2}\bigg)\bE_{t_i}  \int_{t_i}^{t_{i+1}}|\Delta\beta_t|^2 dt \\
	&\le(1+\lambda_1h) |\Delta X_{t_i}|^2 +\bigg( \frac{2h}{\lambda_1}+h^2+\lambda_2h^2\bigg) \big(K|\Delta X_{t_i}|^2+b_y|\Delta Y_{t_i}|^2+b_z|\Delta Z_{t_i}|^2\big)\\
	&\quad +(h+\lambda_1 h) \bigg(\sigma_x|\Delta X_{t_i}|^2+\sigma_y|\Delta Y_{t_i}|^2+\sigma_z|\Delta Z_{t_i}|^2\bigg)  \\
	& \quad+\bigg(1+\frac{4}{\lambda_1 }+\frac{1}{\lambda_2}+\lambda_2 \bigg)\bE_{t_i}\int_{t_i}^{t_{i+1}}\bigg(|\Delta\alpha_t|^2 +|\Delta\beta_t|^2 \bigg) dt,
\end{align*}

	for $\lambda_1, \lambda_2>0$, which completes the proof.
\end{proof}

\section{Proof of Lemma \ref{R4_bw_est}}
\begin{proof}
	 We define the difference
	\begin{equation}
		\Delta f:=f(t_i,V_{t_i},X_{t_i}^1,Y_{t_i}^1,\hat{Z}_{t_i}^1)-f(t_i,V_{t_i},X_{t_i}^2,Y_{t_i}^2,\hat{Z}_{t_i}^2).
	\end{equation}
	Then the backward process difference satisfies
	\begin{equation}
		\Delta Y_{t_i}+\int_{t_i}^{t_{i+1}}\Delta Z_tdW_t=\Delta Y_{t_{i+1}}+\Delta fh+\int_{t_i}^{t_{i+1}}\Delta \gamma_t dt.
	\end{equation}
	Squaring both sides yields
	\begin{align*}
	&	|\Delta Y_{t_i}|^2+\bigg|\int_{t_i}^{t_{i+1}}\Delta Z_tdW_t\bigg|^2+2\Delta Y_{t_i}\cdot \int_{t_i}^{t_{i+1}}\Delta Z_tdW_t\\
		&=\big|\Delta Y_{t_{i+1}}+\Delta fh\big|^2 +2(\Delta Y_{t_{i+1}}+\Delta fh)\int_{t_i}^{t_{i+1}}\Delta \gamma_t dt+\bigg|\int_{t_i}^{t_{i+1}}\Delta \gamma_t dt\bigg|^2\\
		& \leq (1+\lambda_4h)\big|\Delta Y_{t_{i+1}}+\Delta fh\big|^2+\left(1+\frac{1}{\lambda_4 h}\right)\bigg|\int_{t_i}^{t_{i+1}}\Delta \gamma_t dt\bigg|^2.
	\end{align*}
	Then, taking conditional expectations on both sides, we obtain
	\begin{align}\label{ref2}
	&	|\Delta Y_{t_i}|^2+\bE_{t_i}\int_{t_i}^{t_{i+1}}|\Delta Z_t|^2dt\nonumber\\
		& \leq (1+\lambda_4h)\bE_{t_i}\bigg[|\Delta Y_{t_{i+1}}|^2+|\Delta fh|^2+2\Delta Y_{t_{i+1}}\Delta fh\bigg]+\bigg(1+\frac{1}{\lambda_4 h}\bigg)\bE_{t_i}\bigg|\int_{t_i}^{t_{i+1}}\Delta \gamma_t dt\bigg|^2.
	\end{align}
	Note that,
	\begin{align*}
		\bE_{t_i}\int_{t_i}^{t_{i+1}}|\Delta Z_t|^2dt&\geq \frac{1}{h} \bigg|\bE_{t_i}\int_{t_i}^{t_{i+1}}\Delta Z_t dt\bigg|^2\\
		&=\frac{1}{h}\bigg|\bE_{t_i}\bigg[\bigg(\Delta Y_{t_{i+1}}-\Delta Y_{t_i}+\Delta fh+\int_{t_i}^{t_{i+1}}\Delta\gamma_tdt \bigg)\Delta W_{t_i}\bigg] \bigg|^2\\
		&=\frac{1}{h}\bigg|h\Delta \hat{Z}_{t_i}+\bE_{t_i}\bigg[\Delta W_{t_i}\int_{t_i}^{t_{i+1}}\Delta\gamma_tdt \bigg]\bigg|^2\\
		&\geq \frac{1}{h}\bigg(h^2|\Delta\hat{Z}_{t_i}|^2-2h\bigg|\bE_{t_i}\bigg[\Delta\hat{Z}_{t_i}\Delta W_{t_i} \int_{t_i}^{t_{i+1}}\Delta\gamma_tdt\bigg]\bigg|\bigg)\\
		&=h|\Delta\hat{Z}_{t_i}|^2-2\bigg|\bE_{t_i}\bigg[\Delta\hat{Z}_{t_i}\Delta W_{t_i} \int_{t_i}^{t_{i+1}}\Delta\gamma_tdt\bigg]\bigg|\\
		&\geq (1-\lambda_3)h|\Delta\hat{Z}_{t_i}|^2-\frac{1}{\lambda_3}\bE_{t_i}\bigg|\int_{t_i}^{t_{i+1}}\Delta\gamma_tdt\bigg|^2\\
		&\geq (1-\lambda_3)h|\Delta\hat{Z}_{t_i}|^2-\frac{1}{\lambda_3}\bE_{t_i}\bigg|\int_{t_i}^{t_{i+1}}\Delta\gamma_tdt\bigg|^2,\\
	\end{align*}
and 
\begin{align*}
	2h\Delta Y_{t_{i+1}}\Delta f&=2h\Delta Y_{t_{i+1}}\big|f(t_i,V_{t_i},X_{t_i}^1,Y_{t_i}^1,\hat{Z}_{t_i}^1)-f(t_i,V_{t_i},X_{t_i}^2,Y_{t_i}^2,\hat{Z}_{t_i}^2)\big|\\
	&\leq \frac{h}{\lambda_5}\big|\Delta Y_{t_{i+1}}\big|^2+ h\lambda_5\big|f(t_i,V_{t_i},X_{t_i}^1,Y_{t_i}^1,\hat{Z}_{t_i}^1)-f(t_i,V_{t_i},X_{t_i}^2,Y_{t_i}^2,\hat{Z}_{t_i}^2)\big|^2\\
	&\leq\frac{h}{\lambda_5}|\Delta Y_{t_{i+1}}|^2+\lambda_5h\bigg(f_x|\Delta X_{t_i}|^2+K|\Delta Y_{t_i}|^2+f_z|\Delta \hat{Z}_{t_i}|^2\bigg).
\end{align*}
	Using above two estimates in (\ref{ref2}), we finally get
	\begin{align}\label{ref3}
		|\Delta Y_{t_i}|^2+(1-\lambda_3)h|\Delta \hat{Z}_{t_i}|^2 &\leq (1+\lambda_4h)\bE_{t_i}\bigg[|\Delta Y_{t_{i+1}}|^2+|\Delta fh|^2+2\Delta Y_{t_{i+1}}\Delta fh\bigg]\nonumber\\
		&+\left(1+\frac{1}{\lambda_4h }\right)\bE_{t_i}\bigg|\int_{t_i}^{t_{i+1}}\Delta \gamma_t dt\bigg|^2+\frac{1}{\lambda_3}\bE_{t_i}\bigg|\int_{t_i}^{t_{i+1}}\Delta\gamma_tdt\bigg|^2\nonumber\\
 &\leq (1+\lambda_4h)\bE_{t_i}\bigg[|\Delta Y_{t_{i+1}}|^2+h^2\big(f_x|\Delta X_{t_i}|^2+K|\Delta Y_{t_i}|^2+f_z|\Delta \hat{Z}_{t_i}|^2\big)\nonumber\\
		&+\frac{h}{\lambda_5}\bE_{t_i}|\Delta Y_{t_{i+1}}|^2+\lambda_5h\bigg(f_x|\Delta X_{t_i}|^2+K|\Delta Y_{t_i}|^2+f_z|\Delta \hat{Z}_{t_i}|^2\bigg)\bigg]\nonumber\\
		&+\left(1+\frac{1}{\lambda_4h }+\frac{1}{\lambda_3}\right)\bE_{t_i}\bigg|\int_{t_i}^{t_{i+1}}\Delta \gamma_t dt\bigg|^2,
	\end{align}
	which completes the proof.
\end{proof}

\section{Proof of Theorem \ref{Rmain_convergence}}
 \begin{proof}
 	% the proof of (R-main) goes in here....
 	\textbf{Step 1: }From the forward equation of (\ref{maineq}) we can have
 	\begin{align}\label{fw_c}
 		X_{t_{i+1}}&=X_{t_i}+b(t_i,V_{t_i},X_{t_i},Y_{t_i},\hat{Z}_{t_i})h+\int_{t_i}^{t_{i+1}}\Big(b(t,V_t,X_t,Y_t,Z_t)-b(t_i,V_{t_i},X_{t_i},Y_{t_i},\hat{Z}_{t_i})\Big)dt\nonumber\\
 		&+ \sigma(t_i,V_{t_i},X_{t_i},Y_{t_i},\hat{Z}_{t_i})\Delta W_{t_i}+\int_{t_i}^{t_{i+1}}\Big(\sigma(t,V_t,X_t,Y_t,Z_t)-\sigma(t_i,V_{t_i},X_{t_i},Y_{t_i},\hat{Z}_{t_i})\Big)dW_t,
 	\end{align}
 	with $\hat{Z}_{t_i}=\frac{1}{h}\bE_{t_i}\big[ Y_{t_{i+1}}\Delta W_{t_i}\big]$.
 	And from (\ref{discrete_eq}) we have,
 	\begin{align}\label{fw_d}
 		\overline{X}_{t_{i+1}}^{\pi}=\overline{X}_{t_{i}}^{\pi}+b\big(t_i,V_{t_i}\overline{X}_{t_{i}}^{\pi},\overline{Y}_{t_{i}}^{\pi},\overline{Z}_{t_{i}}^{\pi}\big)h+\sigma\big(t_i,V_{t_i}\overline{X}_{t_{i}}^{\pi},\overline{Y}_{t_{i}}^{\pi},\overline{Z}_{t_{i}}^{\pi}\big)\Delta W_{t_{i}}.
 	\end{align}
 	Now applying Lemma \ref{R3_fw_est} on $X$ and $\overline{X}^{\pi}$ from (\ref{fw_c}) and (\ref{fw_d}) respectively, we have
 	\begin{align}\label{fw_err_1}
 		&\bE\big[|X_{t_{i+1}}-\overline{X}_{t_{i+1}}^{\pi}|^2\big]\nonumber\\
 		&\leq (1+A_1h)\bE|X_{t_{i}}-\overline{X}_{t_{i}}^{\pi}|^2+A_2h\bE|Y_{t_{i}}-\overline{Y}_{t_{i}}^{\pi}|^2+A_3h\bE|\hat{Z}_{t_{i}}-\overline{Z}_{t_{i}}^{\pi}|^2\nonumber\\
 		&+C\bE\int_{t_i}^{t_{i+1}}\big|b(t,V_t,X_t,Y_t,Z_t)-b(t_i,V_{t_i},X_{t_i},Y_{t_i},\hat{Z}_{t_i})\big|^2dt   \nonumber\\
 		&+C\bE\int_{t_i}^{t_{i+1}}\big|\sigma(t,V_t,X_t,Y_t,Z_t)-\sigma(t_i,V_{t_i},X_{t_i},Y_{t_i},\hat{Z}_{t_i})\big|^2dt\nonumber\\
 		&\leq (1+A_1h)\bE|X_{t_{i}}-\overline{X}_{t_{i}}^{\pi}|^2+A_2h\bE|Y_{t_{i}}-\overline{Y}_{t_{i}}^{\pi}|^2+A_3h\bE|\hat{Z}_{t_{i}}-\overline{Z}_{t_{i}}^{\pi}|^2\nonumber\\
 		&+C\bE\int_{t_i}^{t_{i+1}}\bigg(\rho(h)+|X_t-X_{t_i}|^2+|Y_t-Y_{t_i}|^2+|Z_t-\hat{Z}_{t_i}|^2\bigg)dt.
 	\end{align}
 	Now since,
 	\begin{align*}
 		&h^2\bE|\tilde{Z}_{t_i}-\hat{Z}_{t_i}|^2\\
 		&=\bE\bigg[\bigg|\bE_{t_i}\bigg[\int_{t_i}^{t_{i+1}}Z_tdt\bigg]-\bE_{t_i}\bigg[\bigg(Y_{t_i}-\int_{t_i}^{t_{i+1}}f(t,V_t,X_t,Y_t,Z_t)dt+\int_{t_i}^{t_{i+1}}Z_tdW_t\bigg)\Delta W_{t_i}\bigg]\bigg|^2\bigg]\\
 		&=\bE\bigg[\bigg|\bE_{t_i}\bigg[\bigg(\int_{t_i}^{t_{i+1}}f(t,V_t,X_t,Y_t,Z_t)dt\bigg)\Delta W_{t_i}\bigg]\bigg|^2\bigg]\\
 		&\leq\bE\left[\bE_{t_i}\bigg|\int_{t_i}^{t_{i+1}}f(t,V_t,X_t,Y_t,Z_t)dt\bigg|^2\cdot\bE_{t_i}\big|\Delta W_{t_i}\big|^2\right]\\
 		&\leq h^2\bE\int_{t_i}^{t_{i+1}}\bigg|f(t,V_t,X_t,Y_t,Z_t)\bigg|^2dt\\
 		&\leq Ch^2\bE\bigg[\int_{t_i}^{t_{i+1}}\Big(|f(t,V_t,0,0,0)|^2+f_x|X_t|^2+K|Y_t|^2+f_z|Z_t|^2\Big)dt\bigg]\\
 		&\leq C\big(I_{0}^2+\bE|x_0|^2\big)h^3 +Ch^2\bE\left[\int_{t_i}^{t_{i+1}}|Z_s|^2ds+\int_{t_i}^{t_{i+1}}|f|^2(s,V_s,0,0,0)ds\right],
 	\end{align*}
 	we can deduce that,
 	\begin{align}\label{est_Zhat}
 		\bE&\int_{t_i}^{t_{i+1}}|Z_t-\hat{Z}_{t_i}|^2dt\leq 2\bE\int_{t_i}^{t_{i+1}}\Big(|Z_t-\tilde{Z}_{t_i}|^2+|\tilde{Z}_{t_i}-\hat{Z}_{t_i}|^2\Big)dt\nonumber\\
 		& =2\bE\int_{t_i}^{t_{i+1}}|Z_t-\tilde{Z}_{t_i}|^2dt +h\bE|\tilde{Z}_{t_i}-\hat{Z}_{t_i}|^2\nonumber\\
 		& \leq C\bE\int_{t_i}^{t_{i+1}}|Z_t-\tilde{Z}_{t_i}|^2dt+\frac{C}{h}\left(\big(I_{0}^2+\bE|x_0|^2\big)h^3 +h^2\bE\bigg[\int_{t_i}^{t_{i+1}}|Z_t|^2dt+\int_{t_i}^{t_{i+1}}|f|^2(t,V_t,0,0,0)dt\bigg]\right)\nonumber\\
 		&\leq C\big(I_{0}^2+\bE|x_0|^2\big)h^2 +C\bE\int_{t_i}^{t_{i+1}}|Z_t-\tilde{Z}_{t_i}|^2dt+Ch\bE\bigg[\int_{t_i}^{t_{i+1}}|Z_t|^2dt+\int_{t_i}^{t_{i+1}}|f|^2(t,V_t,0,0,0)dt\bigg].
 	\end{align}
 	Then, using estimate from Theorem \ref{stability_FBSDE} and equation (\ref{est_Zhat}) into (\ref{fw_err_1}) we have,
 	\begin{align}
 		&\bE\Big[|X_{t_{i+1}}-\overline{X}_{t_{i+1}}^{\pi}|^2\Big]\nonumber\\
 		&\leq \hspace{2mm}(1+A_1h)\bE|X_{t_{i}}-\overline{X}_{t_{i}}^{\pi}|^2+A_2h\bE|Y_{t_{i}}-\overline{Y}_{t_{i}}^{\pi}|^2+A_3h\bE|\hat{Z}_{t_{i}}-\overline{Z}_{t_{i}}^{\pi}|^2\nonumber\\
 		& + C\Big(\big(I_{0}^2+\bE|x_0|^2\big)h+\rho(h)\Big)h+C\bE\int_{t_i}^{t_{i+1}}|Z_t-\tilde{Z}_{t_i}|^2dt\nonumber\\
        &+Ch\bE\bigg[\int_{t_i}^{t_{i+1}}|Z_t|^2dt+\int_{t_i}^{t_{i+1}}\big(|f|^2+|\sigma|^2\big)(t,V_t,0,0,0)dt\bigg].
 	\end{align}
 	Finally, by induction we obtain that, for $1\leq n\leq N$,
 	\begin{align}\label{fw_err}
 		&\bE|X_{t_{n}}-\overline{X}_{t_{n}}^{\pi}|^2\nonumber\\
 		&\leq \hspace{2mm}\sum_{i=0}^{n-1}e^{A_1h(n-i-1)}A_2h\bE|Y_{t_{i}}-\overline{Y}_{t_{i}}^{\pi}|^2+\sum_{i=0}^{n-1}e^{A_1h(n-i-1)}A_3h\bE|\hat{Z}_{t_{i}}-\overline{Z}_{t_{i}}^{\pi}|^2\nonumber\\
 		& + C\sum_{k=0}^{n-1}e^{A_1h(n-k-1)}\bigg(\Big(\big(I_{0}^2+\bE|x_0|^2\big)h+\rho(h)\Big)h+\sup_{0\leq i\leq n}\bE\int_{t_i}^{t_{i+1}}|Z_t-\tilde{Z}_{t_i}|^2dt\nonumber\\
        &+h\sup_{0\leq i\leq n}\bE\bigg[\int_{t_i}^{t_{i+1}}|Z_t|^2dt+\int_{t_i}^{t_{i+1}}\big(|f|^2+|\sigma|^2\big)(t,V_t,0,0,0)dt\bigg]\bigg).
 	\end{align}
 	\textbf{Step 2: }From the backward equation of (\ref{maineq}) we can have
 	\begin{align}\label{bw_c}
 		Y_{t_{i}}=&Y_{t_{i+1}}+f(t_i,V_{t_i},X_{t_i},Y_{t_i},\hat{Z}_{t_i})h-\int_{t_i}^{t_{i+1}}Z_tdW_t\nonumber\\
 		& +\int_{t_i}^{t_{i+1}}\Big(f(t,V_t,X_t,Y_t,Z_t)-f(t_i,V_{t_i},X_{t_i},Y_{t_i},\hat{Z}_{t_i})\Big)dt,
 	\end{align}
 	with $\hat{Z}_{t_i}=\frac{1}{h}\bE_{t_i}\big[ Y_{t_{i+1}}\Delta W_{t_i}\big]$.
 	And from (\ref{discrete_eq}) we have,
 	\begin{align*}
 		\overline{Y}_{t_{i}}^{\pi}=\bE\bigg[\overline{Y}_{t_{i+1}}^{\pi}+f\big(t_i,V_{t_i},\overline{X}_{t_{i}}^{\pi},\overline{Y}_{t_{i}}^{\pi},\overline{Z}_{t_{i}}^{\pi}\big)h\big|\mathcal{F}_{t_i}\bigg].
 	\end{align*}
 	By the martingale representation theorem there exists an $\mathcal{F}_t$-adapted square-integrable process $\{\Bar{Z}_t\}_{t_i\leq t\leq t_{i+1}}$ such that 
 	\begin{align}\label{bw_d}
 		\overline{Y}_{t_{i+1}}^{\pi}&=\bE_{t_i}[\overline{Y}_{t_{i+1}}^{\pi}]+\int_{t_i}^{t_{i+1}}\Bar{Z}_tdW_t=\overline{Y}_{t_{i}}^{\pi}-f\big(t_i,V_{t_i},\overline{X}_{t_{i}}^{\pi},\overline{Y}_{t_{i}}^{\pi},\overline{Z}_{t_{i}}^{\pi}\big)h+\int_{t_i}^{t_{i+1}}\Bar{Z}_tdW_t,\nonumber\\
 		\implies \overline{Y}_{t_i}^{\pi}&=\overline{Y}_{t_{i+1}}^{\pi}+f\big(t_i,V_{t_i},\overline{X}_{t_{i}}^{\pi},\overline{Y}_{t_{i}}^{\pi},\overline{Z}_{t_{i}}^{\pi}\big)h-\int_{t_i}^{t_{i+1}}\Bar{Z}_tdW_t.
 	\end{align}
 	Now applying Lemma \ref{R4_bw_est} on $Y$ and $\overline{Y}^{\pi}$ from (\ref{bw_c}) and (\ref{bw_d}) respectively, we have
 	\begin{align}
 	&	(1-A_6h)\bE|Y_{t_i}-\overline{Y}_{t_i}^{\pi}|^2+A_7h\bE| \hat{Z}_{t_i}-\overline{Z}_{t_i}^{\pi}|^2\nonumber\\
 		&\leq  e^{A_4h}\bE| Y_{t_{i+1}}-\overline{Y}_{t_{i+1}}^{\pi}|^2+A_5h\bE| X_{t_i}-\overline{X}_{t_i}^{\pi}|^2\nonumber\\
 		&+\left(1+\frac{1}{\lambda_4h}+\frac{1}{\lambda_3}\right)\bE\bigg|\int_{t_i}^{t_{i+1}}\Big(f(t,V_t,X_t,Y_t,Z_t)-f(t_i,V_{t_i},X_{t_i},Y_{t_i},\hat{Z}_{t_i})\Big)dt\bigg|^2\nonumber\\
 &\leq  e^{A_4h}\bE| Y_{t_{i+1}}-\overline{Y}_{t_{i+1}}^{\pi}|^2+A_5h\bE| X_{t_i}-\overline{X}_{t_i}^{\pi}|^2\nonumber\\
 &+\left(1+\frac{1}{\lambda_4h}+\frac{1}{\lambda_3}\right)h\bE\int_{t_i}^{t_{i+1}}\big|f(t,V_t,X_t,Y_t,Z_t)-f(t_i,V_{t_i},X_{t_i},Y_{t_i},\hat{Z}_{t_i})\big|^2dt\nonumber\\
 &\leq  e^{A_4h}\bE| Y_{t_{i+1}}-\overline{Y}_{t_{i+1}}^{\pi}|^2+A_5h\bE| X_{t_i}-\overline{X}_{t_i}^{\pi}|^2\nonumber\\
 &+C\bE\int_{t_i}^{t_{i+1}}\Big(\rho(h)+|X_t-X_{t_i}|^2+|Y_t-Y_{t_i}|^2+|Z_t-\hat{Z}_{t_i}|^2\Big)dt\nonumber\\
 &\leq  e^{A_4h}\bE| Y_{t_{i+1}}-\overline{Y}_{t_{i+1}}^{\pi}|^2+A_5h\bE| X_{t_i}-\overline{X}_{t_i}^{\pi}|^2+C\Big(\big(I_{0}^2+\bE|x_0|^2\big)h+\rho(h)\Big)h\nonumber\\
&+C\bE\int_{t_i}^{t_{i+1}}|Z_t-\hat{Z}_{t_i}|^2dt+Ch\bE\bigg[\int_{t_i}^{t_{i+1}}|Z_t|^2dt+\int_{t_i}^{t_{i+1}}|\sigma|^2(t,V_t,0,0,0)dt\bigg].
 %&+Ch\bigg(\Big((I_0^2+\bE|x_0|^2)h+\rho(h)\Big)h+h\bE\int_{t_i}^{t_{i+1}}|Z_t|^2dt+\bE\int_{t_i}^{t_{i+1}}|Z_t-\hat{Z}_{t_i}|^2dt\bigg).
	\end{align}
 	Now using estimate from (\ref{est_Zhat}) we further have
 	\begin{align}
 		&(1-A_6h)\bE|Y_{t_i}-\overline{Y}_{t_i}^{\pi}|^2 +A_7h\bE| \hat{Z}_{t_i}-\overline{Z}_{t_i}^{\pi}|^2\nonumber\\
 		&\leq  e^{A_4h}\bE| Y_{t_{i+1}}-\overline{Y}_{t_{i+1}}^{\pi}|^2+A_5h\bE| X_{t_i}-\overline{X}_{t_i}^{\pi}|^2+C\Big(\big(I_{0}^2+\bE|x_0|^2\big)h+\rho(h)\Big)h\nonumber\\
 		&+C\bE\int_{t_i}^{t_{i+1}}|Z_t-\tilde{Z}_{t_i}|^2dt+Ch\bE\bigg[\int_{t_i}^{t_{i+1}}|Z_t|^2dt+\int_{t_i}^{t_{i+1}}\big(|f|^2+|\sigma|^2\big)(t,V_t,0,0,0)dt\bigg].
 	\end{align}
 	Then, by induction we obtain that, for $0\leq n\leq N-1$,
 	\begin{align}\label{bw_err}
 		&\bE|Y_{t_n}-\overline{Y}_{t_n}^{\pi}|^2 +\sum_{i=n}^{N-1}e^{A_4(i-n)h}A_7h(1-A_6h)^{-(i+1)}\bE| \hat{Z}_{t_i}-\overline{Z}_{t_i}^{\pi}|^2\nonumber\\
 		&\leq  e^{A_4(N-n)h}(1-A_6h)^{-N}\bE| Y_{t_N}-\overline{Y}_{t_{N}}^{\pi}|^2+\sum_{i=n}^{N-1}e^{A_4(i-n)h}A_5h(1-A_6h)^{-(i+1)}\bE| X_{t_i}-\overline{X}_{t_i}^{\pi}|^2\nonumber\\
 		&+C\sum_{i=n}^{N-1}e^{A_4(i-n)h}(1-A_6h)^{-(i+1)}hM_n^{N-1},
 	\end{align}
 	where,
 	\begin{align}
 		M_n^{N-1}=&\big(I_{0}^2+\bE|x_0|^2\big)h+\rho(h)+\frac{1}{h}\sup_{n\leq k\leq N-1}\bE\int_{t_k}^{t_{k+1}}|Z_t-\tilde{Z}_{t_k}|^2dt\nonumber\\
        &+\sup_{n\leq k\leq N-1}\bE\bigg[\int_{t_k}^{t_{k+1}}|Z_t|^2dt+\int_{t_k}^{t_{k+1}}\big(|f|^2+|\sigma|^2\big)(t,V_t,0,0,0)dt\bigg].
 	\end{align}
 	\textbf{Step 3: } Note that $\bE| Y_{t_N}-\overline{Y}_{t_{N}}^{\pi}|^2 = \bE|g(V_T,X_{T})-g(V_T,\overline{X}^{\pi}_{T})|^2\leq g_x|X_T-\overline{X}_T^{\pi}|^2$.\\
 	Now first denote
 	\begin{equation*}
 		\Delta X:=X-\overline{X}^{\pi},\hspace{2mm}\Delta Y:=Y-\overline{Y}^{\pi},\hspace{2mm}\Delta Z:=\hat{Z}-\overline{Z}^{\pi},
 	\end{equation*}
 	and then
 	\begin{align*}
 		P&:=\max_{0\leq n\leq N} e^{-A_1nh}\hspace{2mm} \bE|\Delta X_{t_n}|^2,\\
 		S&:=\max \bigg\{\max_{0\leq n\leq N} e^{A_4nh}\hspace{2mm} \bE|\Delta Y_{t_n}|^2, \sum_{i=0}^{N-1} e^{A_4ih}A_7h(1-A_6h)^{-(i+1)} \hspace{2mm} \bE|\Delta Z_{t_i}|^2\bigg\}.
 	\end{align*}
 	Then from (\ref{fw_err}), we have for $1\leq n\leq N$,
 	\begin{align*}
 		e^{-A_1nh}\bE|\Delta X_{t_n}|^2
 		&\leq \hspace{2mm}\sum_{i=0}^{n-1}e^{-A_1h(i+1)}A_2h\bE|\Delta Y_{t_i}|^2+\sum_{i=0}^{n-1}e^{-A_1h(i+1)}A_3h\bE|\Delta Z_{t_i}|^2+ C\sum_{i=0}^{n-1}e^{-A_1h(i+1)}hM_0^n\nonumber\\
 		& \leq \bigg(\sum_{i=0}^{n-1}e^{-A_1(i+1)h-A_4ih}A_2h+\frac{A_3}{A_7}e^{-A_1h}\bigg)S+C\sum_{i=0}^{n-1}e^{-A_1h(i+1)}hM_0^n\nonumber\\
 		& \leq \bigg(e^{-A_1h}A_2h\frac{1-e^{-(A_1+A_4)T}}{1-e^{-(A_1+A_4)h}}+\frac{A_3}{A_7}e^{-A_1h}\bigg)S+C\sum_{i=0}^{n-1}e^{-A_1h(i+1)}hM_0^n\nonumber\\
 		& = \tilde{A}_1(h)S +C\sum_{i=0}^{n-1}e^{-A_1h(i+1)}hM_0^n,
 	\end{align*}
where
 \begin{align}
 	M_0^{n}=&\big(I_{0}^2+\bE|x_0|^2\big)h+\rho(h)+\frac{1}{h}\sup_{0\leq k\leq n}\bE\int_{t_k}^{t_{k+1}}|Z_t-\tilde{Z}_{t_k}|^2dt\nonumber\\
    &+\sup_{0\leq k\leq n}\bE\bigg[\int_{t_k}^{t_{k+1}}|Z_t|^2dt+\int_{t_k}^{t_{k+1}}\big(|f|^2+|\sigma|^2\big)(t,V_t,0,0,0)dt\bigg],
 \end{align}
and
\begin{equation}
	\tilde{A}_1(h)= e^{-A_1h}A_2h\frac{1-e^{-(A_1+A_4)T}}{1-e^{-(A_1+A_4)h}}+\frac{A_3}{A_7}e^{-A_1h}.
\end{equation}
 	Thus we have,
 	\begin{align}\label{est:p<s}
 		P&\leq  \tilde{A}_1(h)S+C\sum_{i=0}^{N-1}e^{-A_1h(i+1)}hM_0^{N-1}.
 	\end{align}
 	Similarly from (\ref{bw_err}), we have for $0\leq n\leq N-1$,
 	\begin{align}
 		&e^{A_4nh}\bE|\Delta Y_{t_n}|^2 +\sum_{i=n}^{N-1}e^{A_4ih}A_7h(1-A_6h)^{-(i+1)}\bE| \Delta Z_{t_i}|^2\nonumber\\
 		&\leq  e^{A_4Nh}(1-A_6h)^{-N}\bE| \Delta Y_{t_N}|^2+\sum_{i=n}^{N-1}\bigg(e^{A_4ih}A_5h\bE| \Delta X_{t_i}|^2+Ce^{A_4ih}hM_n^{N-1}\bigg)(1-A_6h)^{-(i+1)}\nonumber\\
 		&\leq  e^{A_4T}(1-A_6h)^{-N}g_x\bE| \Delta X_{t_N}|^2+\sum_{i=n}^{N-1}\bigg(e^{A_4ih}A_5h\bE| \Delta X_{t_i}|^2+Ce^{A_4ih}hM_n^{N-1}\bigg)(1-A_6h)^{-(i+1)}\nonumber\\
 		&\leq  e^{A_4T+A_1T}(1-A_6h)^{-N}g_xe^{-A_1T}\bE| \Delta X_{t_N}|^2\nonumber\\
 		&\quad +\sum_{i=n}^{N-1}\bigg(e^{(A_4+A_1)ih}A_5he^{-A_1ih}\bE| \Delta X_{t_i}|^2+Ce^{A_4ih}hM_n^{N-1}\bigg)(1-A_6h)^{-(i+1)}\nonumber\\
 		&\leq  \bigg(e^{A_4T+A_1T}(1-A_6h)^{-N}g_x+\sum_{i=n}^{N-1}e^{(A_4+A_1)ih}A_5h(1-A_6h)^{-(i+1)}\bigg)P\nonumber\\
 		&\quad+C\sum_{i=n}^{N-1}e^{A_4ih}(1-A_6h)^{-(i+1)}hM_n^{N-1}\nonumber \\
 		&\leq  \bigg(e^{(A_4+A_1)T}(1-A_6h)^{-N}g_x+A_5h(1-A_6h)^{-N}\frac{e^{(A_4+A_1)T}-1}{e^{(A_4+A_1)h}-1}\bigg)P\nonumber\\
 		& \quad +(1-A_6h)^{-N}C\sum_{i=0}^{N-1}e^{A_4ih}hM_0^{N-1}.
 	\end{align}
 	And thus we have,
 	\begin{align}\label{est:s<p}
 		S\leq \tilde{A}_2(h)P+(1-A_6h)^{-N}C\sum_{i=0}^{N-1}e^{A_4ih}hM_0^{N-1},
 	\end{align}
 	where,
 	\begin{align}
 		\tilde{A}_2(h)=e^{(A_4+A_1)T}(1-A_6h)^{-N}g_x+A_5h(1-A_6h)^{-N}\frac{e^{(A_4+A_1)T}-1}{e^{(A_4+A_1)h}-1}.
 	\end{align}
 	Now combining inequalities \eqref{est:p<s} and \eqref{est:s<p} we can have
 	\begin{align}\label{final_est_bw}
 		S& \leq \tilde{A}_2(h)\bigg(\tilde{A}_1(h)S+C\sum_{i=0}^{N-1}e^{-A_1h(i+1)}hM_0^{N-1}\bigg)+(1-A_6h)^{-N}C\sum_{i=0}^{N-1}e^{A_4ih}hM_0^{N-1}\nonumber\\
 		&= \tilde{A}_1(h)\tilde{A}_2(h)S+C\sum_{i=0}^{N-1}\bigg(h\tilde{A}_2(h)e^{-A_1h(i+1)}+(1-A_6h)^{-N}e^{A_4ih}h\bigg)M_0^{N-1}.
 	\end{align}
 	Combining inequalities \eqref{est:p<s} and \eqref{est:s<p} we can also have
 	\begin{align}\label{final_est_fw}
 		P&\leq \tilde{A}_1(h)\bigg(\tilde{A}_2(h)P+(1-A_6h)^{-N}C\sum_{i=0}^{N-1}e^{A_4ih}hM_0^{N-1}\bigg)+C\sum_{i=0}^{N-1}e^{-A_1h(i+1)}hM_0^{N-1}\nonumber\\
 		&= \tilde{A}_1(h)\tilde{A}_2(h)P+C\sum_{i=0}^{N-1}\bigg(\tilde{A}_1(h)(1-A_6h)^{-N}e^{A_4ih}h+he^{-A_1h(i+1)}\bigg)M_0^{N-1},
 	\end{align}
 	Note that,
 	\begin{align*}
 		\overline{\tilde{A}}_1= \lim_{h\rightarrow 0}\tilde{A}_1(h)&=\lim_{h\rightarrow 0} \bigg(e^{-A_1h}A_2h\frac{1-e^{-(A_1+A_4)T}}{1-e^{-(A_1+A_4)h}}+\frac{A_3}{A_7}e^{-A_1h}\bigg)\nonumber\\
 		&=\frac{\overline{A}_2(1-e^{-(\overline{A}_1+\overline{A}_4)T})}{\overline{A}_1+\overline{A}_4}+\frac{\overline{A}_3}{\overline{A}_7},
 	\end{align*}
 and
 	\begin{align*}
 		\overline{\tilde{A}}_2=\lim_{h\rightarrow 0}\tilde{A}_2\left(h\right)&=\lim_{h\rightarrow 0}\bigg(e^{(A_4+A_1)T}(1-A_6h)^{-N}g_x+A_5h(1-A_6h)^{-N}\frac{e^{(A_4+A_1)T}-1}{e^{(A_4+A_1)h}-1}\bigg)\nonumber\\
 		&=e^{(\overline{A}_1+\overline{A}_4+\overline{A}_6)T}g_x+ \frac{\overline{A}_5(e^{(\overline{A}_1+\overline{A}_4+\overline{A}_6)T}-1)}{\overline{A}_1+\overline{A}_4},
 	\end{align*}
 	and as a result
 	\begin{align}
 		\lim_{h\rightarrow 0}\tilde{A}_1(h)\tilde{A}_2(h)=\overline{A}_0=\bigg(\frac{\overline{A}_2(1-e^{-(\overline{A}_1+\overline{A}_4)T})}{\overline{A}_1+\overline{A}_4}+\frac{\overline{A}_3}{\overline{A}_7}\bigg) \bigg(e^{(\overline{A}_1+\overline{A}_4+\overline{A}_6)T}g_x+ \frac{\overline{A}_5(e^{(\overline{A}_1+\overline{A}_4+\overline{A}_6)T}-1)}{\overline{A}_1+\overline{A}_4}\bigg).
 	\end{align}
 	Now when $\overline{A}_0<1$, for any $\epsilon>0$ and sufficiently small h, we can have from \eqref{final_est_bw} and \eqref{final_est_fw} that
 	\begin{align}
    \overline{S}&\le (1+\epsilon)[1-\overline{A}_0]^{-1}\bigg(\overline{\tilde{A}}_2\frac{1-e^{-\overline{A}_1T}}{\overline{A}_1}+\frac{e^{\overline{A}_6 T}(e^{\overline{A}_4T}-1)}{\overline{A}_4}\bigg) CM_0^{N-1},\\
 		\overline{P}&\le (1+\epsilon)[1-\overline{A}_0]^{-1}\bigg(\frac{1-e^{-\overline{A}_1T}}{\overline{A}_1}+\overline{\tilde{A}}_1\frac{e^{\overline{A}_6 T}(e^{\overline{A}_4T}-1)}{\overline{A}_4}\bigg)CM_0^{N-1},
 	\end{align}
 	where, 
 	\begin{align*}
 		\overline{P}&=\max_{0\leq n\leq N} e^{-\overline{A}_1nh}\hspace{2mm} \bE|\Delta X_n|^2,\\
 		\overline{S}&=\max \bigg\{\max_{0\leq n\leq N} e^{\overline{A}_4nh}\hspace{2mm} \bE|\Delta Y_n|^2, \sum_{i=0}^{N-1} e^{\overline{A}_4ih}e^{\overline{A}_6T}\overline{A}_7h \hspace{2mm} \bE|\Delta Z_i|^2\bigg\}.
 	\end{align*}
 	By fixing $\epsilon=1$, we obtain our error of $\bE|\Delta X_n|^2, \bE|\Delta Y_n|^2$ and $\bE|\Delta Z_n|^2$ as
 	\begin{equation*}
 		\max_{0\le n\le N}\bE|\Delta X_n|^2\le e^{\overline{A}_1T\vee 0}\overline{P}\le CM_0^{N-1},
 	\end{equation*}
 and
 	\begin{equation*}
 		\max\left\{ \max_{0\le n\le N}\bE|\Delta Y_n|^2,\,\,
 		\sum_{i=0}^{N-1}\bE|\Delta Z_i|^2h\right\}
 		\le e^{-\overline{A}_4T\vee 0}(1+\overline{A}_7^{-1})\overline{S}
 		\le CM_0^{N-1}.
 	\end{equation*}
 	Thus,
 	\begin{equation}
 		\begin{split}
 			\sup_{t\in[0,T]} \bE\bigg[|X_t-\overline{X}_{t}^{\pi}|^2+|Y_t-\overline{Y}_{t}^{\pi}|^2\bigg]+\bE\int_{0}^{T}|\hat{Z}_t-\overline{Z}_{t}^{\pi}|^2dt&\leq CM.
 		\end{split}
 	\end{equation}
 \end{proof}
 
%\newpage
%\subsection{Example 2(Modified): Utility Maximization}

%\subsubsection{}

\end{document}